\documentclass[journal,10pt,draftclsnofoot,onecolumn]{IEEEtran}
\usepackage{amsmath,graphicx,amssymb}
\usepackage[english]{babel}
\usepackage{cite}

\usepackage{threeparttable}
\usepackage{float}
\floatstyle{plaintop}
\restylefloat{table}

\floatstyle{plain}
\restylefloat{figure}
\usepackage{placeins}

\usepackage[linesnumberedhidden,ruled,vlined]{algorithm2e}

\newtheorem{theorem}{Theorem}[section]

\newtheorem{lemma}[theorem]{Lemma}

\numberwithin{equation}{section}

\newcommand{\qed}{\rule{7pt}{7pt}}
\newenvironment{proof}{\noindent{\bf Proof}\hspace*{1em}}{\hfill\qed\vspace{0.125in}}

\begin{document}

\title{Precisely Verifying the Null Space Conditions in Compressed Sensing: A Sandwiching Algorithm}
%Change to Capital letters for initials, probably change .cls or .sty file
% Single address.
% ---------------
\author{Myung Cho~\IEEEmembership{}
        and Weiyu Xu~\IEEEmembership{}%
        %and~Jane~Doe,~\IEEEmembership{}% <-this % stops a space
\thanks{Myung Cho and Weiyu Xu  are with the Department
of Electrical and Computer Engineering, University of Iowa, Iowa City,
IA, 52242 USA e-mail: myung-cho@uiowa.edu, weiyu-xu@uiowa.edu.}% <-this % stops a space
%\thanks{Manuscript received April 19, 2005; revised January 11, 2007.}
}

\maketitle

%\ninept
%
%\maketitle
%
\begin{abstract}
The null space condition of sensing matrices plays an important role in guaranteeing the success of compressed sensing. In this paper, we propose new efficient algorithms to verify the null space condition in compressed sensing (CS). Given an $(n-m) \times n$ ($m>0$) CS matrix $A$ and a positive $k$, we are interested in computing $\displaystyle \alpha_k = \max_{\{z: Az=0,z\neq 0\}}\max_{\{K: |K|\leq k\}}$ $\frac{\|z_K \|_{1}}{\|z\|_{1}}$, where $K$ represents subsets of $\{1,2,...,n\}$, and $|K|$ is the cardinality of $K$. In particular, we are interested in finding the maximum $k$ such that $\alpha_k < \frac{1}{2}$. However, computing $\alpha_k$ is known to be extremely challenging.  In this paper, we first propose a series of new polynomial-time algorithms to compute upper bounds on $\alpha_k$. Based on these new polynomial-time algorithms, we further design a new sandwiching algorithm, to compute the \emph{exact} $\alpha_k$ with greatly reduced complexity. When needed, this new sandwiching algorithm also achieves a smooth tradeoff between computational complexity and result accuracy. Empirical results show the performance improvements of our algorithm over existing known methods; and our algorithm outputs precise values of $\alpha_k$, with much lower complexity than exhaustive search.
\end{abstract}
\section{Introduction}
\label{Intro}

In compressed sensing, a sensing matrix $A \in \mathbf R^{(n-m) \times n}$ with $0<m<n$ is given,
and we have $y=Ax$, where $y \in \mathbf R^{n-m}$ is a measurement result and $x \in \mathbf R^n$ is a signal.
The sparest solution $x$ to the underdetermined equation $y=Ax$ is given by (\ref{Intro_eqn1}):
\begin{align}
\label{Intro_eqn1}
          \min              & \quad \| x \|_{0}       \nonumber \\
          \text{subject to} & \quad Ax = y
\end{align}
When the vector $x$ has only $k$ nonzero elements ($k$-sparse signal, $k \ll n$), the solution of (\ref{Intro_eqn2}),
which is called $\ell_1$ minimization, coincides with the solution of (\ref{Intro_eqn1}) under certain conditions,
such as restricted isometry conditions \cite{candes2005decoding,d2011testing,candes2006robust,candes2006stable,donoho2005neighborly,recht2011null}.
\begin{align}
\label{Intro_eqn2}
          \min              & \quad \| x \|_{1}       \nonumber \\
          \text{subject to} & \quad Ax = y
\end{align}
In order to guarantee that we can recover the sparse signal by solving $\ell_1$ minimization, we need to check these conditions are satisfied. The necessary and sufficient condition for the solution of (\ref{Intro_eqn2}) to coincide with the solution of (\ref{Intro_eqn1}) is the null space condition (NSC) \cite{juditsky2011verifiable,cohen2009compressed}.
Namely, when the NSC holds for a number $k$, then any $k$-sparse signal $x$ can be exactly recovered by solving $\ell_1$ minimization.
This NSC is defined as follows.

Given a matrix $A \in \mathbf R^{(n-m) \times n}$ with $0<m<n$,
\begin{align}
\label{Intro_eqn3}
     & || z_K ||_1 < || z_{\overline K} ||_1, \\
     & \forall z \in \{z:Az=0, z\neq 0\},\; \forall K \subseteq \{1,2,...n\}\; with\; |K|\leq k, \nonumber
\end{align}
where $K$ is an index set, $|K|$ is the cardinality of $K$, $z_K$ is the elements of $z$ vector corresponding to the index set $K$, and $\overline{K}$ is the complement of $K$.
$\alpha_k$ is defined as below, and $\alpha_k$ should be smaller than $\frac{1}{2}$ in order to satisfy the NSC.
\begin{align*}
    \alpha_k = \max_{\{z: Az=0,z\neq 0\}}\max_{\{K: |K|\leq k\}} \frac{\|z_K \|_{1}}{\|z\|_{1}}
\end{align*}
A smaller $\alpha_k$ generally means more robustness in recovering approximately sparse signal $x$ via $\ell_1$ minimization \cite{juditsky2011verifiable,cohen2009compressed,xu2011precise}.

When a matrix $H \in \mathbf R^{n \times m},\; n > m$, is the basis of the null space of $A$ ($AH = 0$), then the property (\ref{Intro_eqn3}) is equivalent to the following property (\ref{Intro_eqn4}):
\begin{align}
\label{Intro_eqn4}
        & \| (Hx)_{K} \|_{1} < \| (Hx)_{\overline{K}} \|_{1}, \\
        & \forall x \in \mathbf{R}^{m}, x \neq 0, \forall K \subseteq \{1,2,...n\}\; with\; |K|\leq k, \nonumber
\end{align}
where $K$ is an index set, $|K|$ is the cardinality of $K$, $(Hx)_K$ is the elements of $(Hx)$ corresponding to the index set $K$, and $\overline{K}$ is the complement of $K$.
(\ref{Intro_eqn4}) holds if and only if the optimum value of (\ref{Intro_eqn6}) is smaller than 1.
We define the optimum value of (\ref{Intro_eqn6}) as $\beta_k$:
\begin{align}
\label{Intro_eqn6}
          \beta_k = \max_{ x \in \mathbf{R}^{m},\; |K| \leq k }  & \quad \| (Hx)_{K} \|_{1}                    \nonumber \\
                    \text{subject to}                            & \quad \| (Hx)_{\overline{K}} \|_{1} \leq 1 .
\end{align}
And then $\alpha_k$ is rewritten as below:
\begin{align*}
    \alpha_k = \max_{\{x \in R^m, x \neq 0\}}\max_{\{K: |K|\leq k\}} \frac{\|(Hx)_K \|_{1}}{\|(Hx)\|_{1}} = \frac{ \beta_k }{ 1+\beta_k }.
\end{align*}
We are interested in computing $\alpha_k$, and particularly finding the maximum $k$ such that $\alpha_k < \frac{1}{2}$.

However, solving the programming (\ref{Intro_eqn6}) is difficult,
because there are at least $\binom{n}{k}$ subsets $K$, which can be exponentially large in $n$ and $k$,
and the objective function is not a concave function.  In fact, \cite{pfetsch2012computational} shows that given a matrix $A$ and a number $k$, computing $\alpha_k$ is strongly NP-hard. Under these computational difficulties, testing the NSC was often conducted by obtaining an upper or lower bound on $\alpha_k$ \cite{d2011testing,juditsky2011verifiable,tang2011performance,tang2011verifiable,lee2008computing}.
In \cite{d2011testing} and \cite{tang2011verifiable}, semidefinite relaxation methods were introduced by transforming the NSC into semidefinite programming to obtain the bounds on $\alpha_k$ or related quantities.
In \cite{juditsky2011verifiable} and \cite{tang2011performance}, linear programming (LP) relaxations were introduced to obtain the upper and lower bounds on $\alpha_k$.
Those papers showed computable performance guarantees on sparse signal recovery with bounds on $\alpha_k$.
However, the bounds resulting from \cite{d2011testing,juditsky2011verifiable,tang2011performance,tang2011verifiable,lee2008computing},
did not provide the exact value of $\alpha_k$, which led to a small $k$ value satisfying the null space conditions.
%%%%%%%%%%%%%%%%%%%%%%%%%%%%%%%%%%%%%%%%% in detail %%%%%%%%%%%%%%%%%%%%%%%%%%%%%%%%%%%%%%%%%%%

%%%%%%%%%%%%%%%%%%%%%% our contribution
In this paper, we first propose a series of new polynomial-time algorithms to compute upper bounds on $\alpha_k$.
Based on these new polynomial-time algorithms, we further design a new sandwiching algorithm, to compute the \emph{exact} $\alpha_k$ with greatly reduced complexity.
This new sandwiching algorithm also offers a natural way to achieve a smooth tradeoff between computational complexity and result accuracy.
By computing the exact $\alpha_k$, we obtained bigger $k$ values than results from \cite{d2011testing} and \cite{juditsky2011verifiable}.
%%%%%%%%%%%%%%%%%%%%%

This paper is organized as follows.
In Section \ref{Sec2}, we provide the pick-$1$-element algorithm and a proof showing that the pick-$1$-element algorithm provides an upper bound on $\alpha_k$.
In Section \ref{Sec3}, we provide the pick-$l$-element algorithms, $2 \leq l \leq k$, and a proof showing that the pick-$l$-element algorithms also provide upper bounds on $\alpha_k$.
In Section \ref{Sec3-1}, we consider the pick-$l$-element algorithms with optimized coefficients, $1 \leq l \leq k$, and a proof showing that when $l$ is increased, upper bound on $\alpha_k$ from the pick-$l$-element algorithm with optimized coefficients becomes smaller or stays the same.
In Section \ref{Sec4}, we propose a sandwiching algorithm
based on the pick-$l$-element algorithms to obtain the exact $\alpha_k$.
In Section \ref{Sec5} and \ref{Sec6}, we provide empirical results showing that the improved performance of our algorithm over existing methods
and conclude our paper by discussing extensions and future directions.

\section{Pick-1-element Algorithm}
\label{Sec2}

Given a matrix $H \in \mathbf R^{n \times m},\; 0<m<n$, in order to verify $\alpha_k < \frac{ 1 }{ 2 }$, we propose a polynomial-time algorithm to find an upper bound on $\alpha_k$. Let us define $\alpha_{1,\{i\}}$ as $\frac{\beta_{1,\{i\}}}{1+\beta_{1,\{i\}}}$ and $\beta_{1,\{i\}}$ as below:
\begin{align}
\label{Sec2_eqn1}
          \beta_{1,\{i\}} = \max_{ x \in \mathbf{R}^{m} }  & \quad \| (Hx)_{ \{i\} } \|_{1}          & \nonumber \\
                        \text{subject to}              & \quad \| (Hx)_{\overline{ \{i\} }} \|_{1} \leq 1,  &
\end{align}
where $(Hx)_{\{i\} }$ is the $i$-th element in $(Hx)$ and $(Hx)_{\overline{ \{i\} }}$ is the rest elements in $(Hx)$.
The subscript $1$ in $\beta_{1,\{i\}}$ is used to represent one element and the $\{i\}$ in $\beta_{1,\{i\}}$ is used to represent the $i$-th element in $(Hx)$. The pick-$1$-element algorithm is given as follows to compute an upper bound on $\alpha_k$. \\

%%%%%%%%%%%%%%%%%%%%%%%%%%%%%%%%%%%%%%%%%%%%%%%%%%%%%%%%%%
%pick-l-element algorithm - pseudo-code
%%%%%%%%%%%%%%%%%%%%%%%%%%%%%%%%%%%%%%%%%%%%%%%%%%%%%%%%%%
%{
\begin{algorithm}
\LinesNumbered
  \caption{Pick-$1$-element Algorithm for computing an upper bound on $\alpha_k$ in Pseudo code}

    \KwIn{H matrix}
    \For {$i=1$ \KwTo $n$} {
        $\beta_{1,\{i\}}$ $\leftarrow$ output of (\ref{Sec2_eqn1}) \par
        $\alpha_{1,\{i\}}  \leftarrow  \beta_{1,\{i\}}/ (1 + \beta_{1,\{i\}})$ \par
    }
    Sort $\alpha_{1,\{i\}}$, $i=1,...,n$, in descending order: $\alpha_{1,\{i_j\}}$, $j=1,...,n$  \par
    Compute an upper bound from the following equation
    \begin{align*}
        \sum_{ j=1 }^{k} \alpha_{1,\{i_j\}}
    \end{align*} \par

    \If {upper bound $< \frac{1}{2}$} {
        NSC is satisfied
    }

\end{algorithm}
%}%%%%%%%%%%%%%%%%%%%%%%%%%%%%%%%%%%%%%%%%%%%%%%%%%%%%%%%%%%%

%%%%%%%%%%%%%%%%%%%%%%%%%%%%%%%%%%%%%%%%%%%%%%%%%%%%%%%%%%
%pick-1-element algorithm
%%%%%%%%%%%%%%%%%%%%%%%%%%%%%%%%%%%%%%%%%%%%%%%%%%%%%%%%%%%
%{
\begin{algorithm}
\LinesNumbered
  \caption{Pick-$1$-element Algorithm for computing an upper bound on $\alpha_k$ in description}

        \nlset{1}Given a matrix $H$, find an optimum value of (\ref{Sec2_eqn1}): $\beta_{1,\{i\}}$, $i=1,2,...,n$. \par
        \nlset{2}Compute $\alpha_{1,\{i\}}$ with the values from Step 1: $\alpha_{1,\{i\}} = \frac{ \beta_{1,\{i\}} }{1+\beta_{1,\{i\}} }$, $i=1,2,...,n$ \par
        \nlset{3}Sort these $n$ different values of $\alpha_{1,\{i\}}$ in descending order: $\alpha_{1,\{i_1\}}, \alpha_{1,\{i_2\}},..., \alpha_{1,\{i_n\}}$, where $\alpha_{1,\{i_1\}} \geq \alpha_{1,\{i_2\}} \geq ... \geq \alpha_{1,\{i_n\}}$ \par
        \nlset{4}Compute the sum of the first $k$ values of $\alpha_{1,\{i_j\}}$:  $\sum_{ j=1 }^{k} \alpha_{1,\{i_j\}} $ \par
        \nlset{5}If the result from Step 4 is smaller $\frac{1}{2}$, then the null space condition is satisfied.
\end{algorithm}
%}%%%%%%%%%%%%%%%%%%%%%%%%%%%%%%%%%%%%%%%%%%%%%%%%%%%%%%%%%%%

\begin{lemma}
\label{lemma2_1}
$\alpha_k$ can not be larger than the sum of the k largest $\alpha_{1,\{i\}}$. Namely,
\begin{align*}
        \alpha_k
        & \leq
        \sum_{ j=1 }^{k} \alpha_{1,\{i_j\}},
\end{align*}
where $\alpha_{1,\{i_1\}} \geq \alpha_{1,\{i_2\}} \geq ... \geq \alpha_{1,\{i_k\}} \geq ... \geq \alpha_{1,\{i_n\}}$, $i_1, i_2,..., i_k,...,i_n \in \{1,2,...,n\}$, and $i_a \neq i_b$ if $a \neq b $. The subscript $j$ of $i_j$ in $\alpha_{1,\{i_j\}}$ is used to represent that the values are sorted.
\end{lemma}

\begin{proof}
We assume that when $x = x^{i}, \;\; i = 1,2,3,...,n$, we achieve the optimum value $\alpha_{1,\{i\}} (= \frac{\beta_{1,\{i\}}}{1+\beta_{1,\{i\}}})$.
Namely,
\begin{align*}
          \beta_{1,\{i\}} = \max_{ x \in \mathbf{R}^{m} }  & \quad \| (Hx)_{ \{i\} } \|_{1}                    & \\
                            \text{subject to}              & \quad \| (Hx)_{\overline{ \{i\} }} \|_{1} \leq 1  &
\end{align*}
And we assume that when $x = x^{*}$, we achieve the optimum value $\alpha_{k} (= \frac{\beta_{k}}{1+\beta_{k}})$.
\begin{align*}
          \beta_k = \max_{ x \in \mathbf{R}^{m},\; |K| \leq k }  & \quad \| (Hx)_{K} \|_{1}                   & \\
                    \text{subject to}                            & \quad \| (Hx)_{\overline{K}} \|_{1} \leq 1 &
\end{align*}
The inequality in Lemma \ref{lemma2_1} is the same as the following (\ref{Sec2_eqn2}):
\begin{align}
\label{Sec2_eqn2}
       \underbrace{ \frac{ \| (Hx^{*})_{K} \|_{1} }{ \| (Hx^{*}) \|_{1} } }_{ |K| \leq k}
       \leq
       \sum_{j=1}^{k} \frac{ \| (Hx^{i_j})_{ \{i_j\} } \|_{1} }{ \| (Hx^{i_j}) \|_{1} },
\end{align}
(\ref{Sec2_eqn2}) can be rewritten as (\ref{Sec2_eqn3}).
\begin{align}
\label{Sec2_eqn3}
       \sum_{i \in K } \frac{ \| (Hx^{*})_{\{i\}} \|_{1} }{ \| (Hx^{*}) \|_{1} }
       \leq
       \sum_{j=1}^{k} \frac{ \| (Hx^{i_j})_{ \{i_j\} } \|_{1} }{ \| (Hx^{i_j}) \|_{1} }
\end{align}
The left-hand side of (\ref{Sec2_eqn3}) can not be larger than the sum of the $\alpha_{1,\{i\}}$, which is the maximum value for the $i$-th element.
\begin{align*}
     & \sum_{i \in K } \frac{ \| (Hx^{*})_{\{i\}} \|_{1} }{ \| (Hx^{*}) \|_{1} } \leq
        \sum_{i \in K } \underbrace{ \frac{ \| (Hx^{i})_{\{i\}} \|_{1} }{ \| (Hx^{i}) \|_{1} } }_{ \substack{ \text{$\alpha_{1,\{i\}}$} \\ \text{maximum value of the $i$-th element} } }
\end{align*}
The sum of $\alpha_{1,\{i\}},\; i \in K$, can not be larger than the sum of the $k$ largest $\alpha_{1,\{i_j\}},\; j=1,2,...,k$.
{\small
\begin{align*}
     &   \sum_{i \in K } \underbrace{ \frac{ \| (Hx^{i})_{\{i\}} \|_{1} }{ \| (Hx^{i}) \|_{1} } }_{ \substack{ \text{maximum value} \\ \text{of $1$ element in a set $K$}  } } \nonumber \\
     &  \leq
       \underbrace{ \frac{ \| (Hx^{i_1})_{ \{i_1\} } \|_{1} }{ \| (Hx^{i_1}) \|_{1} } }_{ \substack{\text{1st max. value of $1$ element}} }
      + ...
     + \underbrace{ \frac{ \| (Hx^{i_k})_{ \{i_k\} } \|_{1} }{ \| (Hx^{i_k}) \|_{1} } }_{ \substack{\text{$k$-th max. value of $1$ element}} }
\end{align*}
}%
\end{proof}

\section{Pick-$l$-element Algorithms}
\label{Sec3}

In order to obtain better bounds on $\alpha_k$ than the pick-$1$-element algorithm,
in this section we generalize the pick-$1$-element algorithm to the pick-$l$-element algorithms,
where $l\geq 2$ is a fixed chosen integer no bigger than $k$.
The basic idea is to first compute the maximum portion $\displaystyle \max_{ x \in \mathbf{R}^{m} }\frac{\|(Hx)_{L}\|_1}{\|(Hx)\|_1}$ for every subset $L \subseteq \{1,2,..,n\}$ with cardinality $|L|=l$.
One can then garner this information to efficiently compute an upper bound on $\alpha_{k}$.

We first index the $\binom{n}{l}$ subsets with cardinality $l$ by indices $1$,$2$,..., and $\binom{n}{l}$; and we denote the subset corresponding to index $i$ as $L_i$.
Let us define $\beta_{l,L_i},\;\; i \in \{1,2,3,..., {\binom{n}{l}} \},\;$ as:
\begin{align}
\label{Sec3_eqn1}
          \beta_{l,L_i}=\max_{ x \in \mathbf{R}^{m} }  & \quad \| (Hx)_{L_{i}} \|_{1}      & \nonumber \\
                        \text{subject to}              & \quad \| (Hx)_{\overline{L_{i}}} \|_{1} \leq 1                      &
\end{align}
The subscript $l$ in $\beta_{l,L_i}$ is used to denote the cardinality $l$ of the set $L_i$,
and $i$ in $\beta_{l,L_i}$ is the index of $L_i$. The pick-$l$-element algorithm in pseudocode and in description are respectively listed as follows. \\

%%%%%%%%%%%%%%%%%%%%%%%%%%%%%%%%%%%%%%%%%%%%%%%%%%%%%%%%%%
%pick-l-element algorithm - pseudo-code
%%%%%%%%%%%%%%%%%%%%%%%%%%%%%%%%%%%%%%%%%%%%%%%%%%%%%%%%%%%
%{
\begin{algorithm}
\LinesNumbered
  \caption{Pick-$l$-element Algorithms, $2 \leq l \leq k$ for computing upper bounds on $\alpha_k$ in Pseudo code}

    \KwIn{H matrix} \par

    \For {$i=1$ \KwTo $n \choose l$}{
        $\beta_{l,L_i}$ $\leftarrow$ output of (\ref{Sec3_eqn1}) \par
        $\alpha_{l,L_i} \leftarrow  \beta_{l,L_i}/ (1 + \beta_{l,L_i})$ \par
    }
    Sort $\alpha_{l,L_i}$, $i=1,...,\binom{n}{l}$ in descending order: $\alpha_{l,L_{i_j}}$, $j=1,...,{\binom{n}{l}}$ \par
    Compute an upper bound from the following equation
    \begin{align*}
        \bigg( \frac{1}{{k-1 \choose l-1}} \bigg) \bigg( \sum_{ j=1 }^{k \choose l} \alpha_{l,L_{i_j}} \bigg)
    \end{align*} \par

    \If {upper bound $< \frac{1}{2}$} {
        NSC is satisfied
    }

\end{algorithm}
%}%%%%%%%%%%%%%%%%%%%%%%%%%%%%%%%%%%%%%%%%%%%%%%%%%%%%%%%%%%%

%%%%%%%%%%%%%%%%%%%%%%%%%%%%%%%%%%%%%%%%%%%%%%%%%%%%%%%%%%
%pick-l-element algorithm
%%%%%%%%%%%%%%%%%%%%%%%%%%%%%%%%%%%%%%%%%%%%%%%%%%%%%%%%%%%
%{
\begin{algorithm}
\LinesNumbered
  \caption{Pick-$l$-element Algorithms, $2 \leq l \leq k$ for computing upper bounds on $\alpha_k$ in description}
        \nlset{1}Given a matrix $H$, find an optimum value of (\ref{Sec3_eqn1})
                : $\beta_{l,L_i},\; i \in \{1,2,..., {n \choose l} \}$. \par
        \nlset{2}Compute $\alpha_{l,L_{i}}$ from $\beta_{l,L_{i}}$
                : $\alpha_{l,L_{i}} = \frac{ \beta_{l,L_{i}} }{ 1+\beta_{l,L_{i}} }$, $i=1,2,..., {n \choose l} $. \par
        \nlset{3}Sort these $n \choose l$ different values of $\alpha_{l,L_{i}}$ in descending order
                : $\alpha_{l,L_{i_j}}$, where $ j=1,2,..., {n \choose l} $ and $\alpha_{l,L_{i_a}} \geq \alpha_{l,L_{i_b}}$ when $a \leq b $. \par
        \nlset{4}Compute the sum of the first ${\binom{n}{l}}$ values of $\alpha_{l,L_{i_j}}$  and
                divide the sum with $( {\binom{k-1}{l-1}} )$:
                \begin{align*}
                    \bigg( \frac{1}{{\binom{k-1}{l-1}}} \bigg) \bigg( \sum_{ j=1 }^{k \choose l} \alpha_{l,L_{i_j}} \bigg)
                \end{align*} \par
        \nlset{5}If the result from Step 4 is smaller than $\frac{1}{2}$, then the null space condition is satisfied. \par
\end{algorithm}
%}%%%%%%%%%%%%%%%%%%%%%%%%%%%%%%%%%%%%%%%%%%%%%%%%%%%%%%%%%%%

The following lemma establishes an upper bound on $\alpha_k$.

\begin{lemma}
\label{lemma3_1}
$\alpha_k$ can not be larger than the output of the pick-$l$-element algorithms, where $2 \leq l \leq k$. Namely,
\begin{align*}
       \alpha_k
        \leq
        \underbrace{  \bigg( \frac{1}{{k-1 \choose l-1}} \bigg) \bigg( \sum_{ j=1 }^{k \choose l} \alpha_{l,L_{i_j}} \bigg) }_{  \substack{ \text{upper bound calculated  with}\\ \text{the pick-$l$-element algorithm}} },
\end{align*}
where $i_j \in \{1,2,3,...,{n \choose l} \}$ ($1\leq j \leq \binom{n}{l}$) are $\binom{n}{l}$ distinct numbers; and $\alpha_{l,L_{i_1}} \geq \alpha_{l,L_{i_2}} \geq ... \geq \alpha_{l,L_{i_{n \choose l}}}$.
\end{lemma}

\begin{proof}
Suppose that the maximum value of the programming (\ref{Intro_eqn6}), namely $\beta_k$, is achieved when $K = K^{*}$. Let $L^{*}_{i}$, $1\leq i \leq \binom{k}{l}$,
be the family of subsets of $K^{*}$, with cardinality $l$. It is not hard to see that each element of $K^{*}$ appears in ${ k-1 \choose l-1 }$ such subsets. In particular, we have
\begin{align*}
        K^{*} = \bigcup_{i=1}^{ k \choose l } L^{*}_i .
\end{align*}
Thus, $\forall x \in \mathbf R^m$, we can represent $\frac{ \| (Hx)_{K^{*}} \|_{1} }{ \| (Hx) \|_{1} }$ as follows.
\begin{align}
\label{Sec3_eqn2}
       \frac{ \| (Hx)_{K^{*}} \|_{1} }{ \| (Hx) \|_{1} } =
       \bigg( \frac{1}{{k-1 \choose l-1}} \bigg) \bigg( \sum_{i=1}^{k \choose l} \frac{ \| (Hx)_{ L^{*}_i  } \|_{1} }{ \| (Hx) \|_{1} } \bigg)
\end{align}
Suppose that each term of the right-hand side of (\ref{Sec3_eqn2}), $\frac{ \| (Hx)_{ L^{*}_i  } \|_{1} }{ \| (Hx) \|_{1} }$, achieves the maximum value when $x = x^{i^{*}}, \;\; i = 1,...,{k \choose l}$;
and the maximum value of $\frac{\|(Hx)_{K^{*}}\|_{1}}{\|Hx\|_1}$ in (\ref{Sec3_eqn2})  is achieved when $x = x^{*}$.
Then, $\forall x \in \mathbf R^m$, we have
\begin{align}
\label{Sec3_eqn3}
       \frac{ \| (Hx)_{K^{*}} \|_{1} }{ \| (Hx) \|_{1} }
       & = \bigg( \frac{1}{{k-1 \choose l-1}} \bigg) \bigg( \sum_{i=1}^{k \choose l} \frac{ \| (Hx)_{ L^{*}_i  } \|_{1} }{ \| (Hx) \|_{1} } \bigg), \; \forall x \in \mathbf R^m \nonumber \\
       & \leq \bigg( \frac{1}{{k-1 \choose l-1}} \bigg) \bigg( \sum_{i=1}^{k \choose l} \frac{ \| (Hx^{{i^{*}}})_{ L^{*}_i  } \|_{1} }{ \| (Hx^{{i^{*}}}) \|_{1} } \bigg).
\end{align}
In the meantime, the maximum output from the  pick-$l$-element algorithm is
\begin{align*}
       &\bigg( \frac{1}{{k-1 \choose l-1}} \bigg) \bigg( \sum_{j=1}^{k \choose l} \frac{ \| (Hx^{i_j})_{ L_{i_j} } \|_{1} }{ \| (Hx^{{i_j}}) \|_{1} } \bigg), j = 1,...,{ k \choose l }.
\end{align*}
By our definitions of indices $i_j$'s, we have
\begin{align}
\label{Sec3_eqn4}
       &\bigg( \frac{1}{{k-1 \choose l-1}} \bigg) \bigg( \sum_{i=1}^{k \choose l} \frac{ \| (Hx^{i^{*}})_{ L^{*}_i  } \|_{1} }{ \| (Hx^{i^{*}}) \|_{1} } \bigg) \nonumber\\
       &\leq \bigg( \frac{1}{{k-1 \choose l-1}} \bigg) \bigg( \sum_{j=1}^{k \choose l} \frac{ \| (Hx^{i_j})_{ L_{i_j}  } \|_{1} }{ \| (Hx^{i_j}) \|_{1} } \bigg).
\end{align}
Combining (\ref{Sec3_eqn3}), and (\ref{Sec3_eqn4}) leads to
\begin{align*}
       & \frac{ \| (Hx^{*})_{K} \|_{1} }{ \| (Hx^{*}) \|_{1} }
         \leq \bigg( \frac{1}{{k-1 \choose l-1}} \bigg) \bigg( \sum_{j=1}^{k \choose l} \frac{ \| (Hx^{i_j})_{ L_{i_j}  } \|_{1} }{ \| (Hx^{i_j}) \|_{1} } \bigg). \nonumber
\end{align*}
Therefore, we have finished proving this lemma.
\end{proof}

\section{Pick-$l$-element algorithms with optimized coefficients}
\label{Sec3-1}

The pick-$l$-element algorithm has $\frac{1}{\binom{k-1}{l-1}}$ as its coefficients. In this section, we show that we can actually strengthen the upper bounds of the pick-$l$-element algorithms, at the cost of additional polynomial-time complexity. In fact, we can calculate improved upper bounds on $\alpha_k$, using \emph{the pick-$l$-element algorithm with optimized coefficients}.
From this new algorithm, we can show when $l$ is increased, the upper bound on $\alpha_k$ becomes smaller or stays the same.

The upper bound from the pick-$l$-element algorithm is given by
\begin{align}
\label{pickl_element_algo}
       \bigg( \frac{1}{{k-1 \choose l-1}} \bigg) \bigg( \sum_{ j=1 }^{k \choose l} \alpha_{l,L_{i_j}} \bigg),
\end{align}
where $ \alpha_{l,L_{i_j}}, j=1,...,\binom{n}{l}$ are sorted in descending order.

The upper bound from the pick-$l$-element algorithm with optimized coefficients is obtained by solving the following problem:
\begin{align}
\label{pickl_optimized_coefficients}
        \max_{\gamma_i,\; 1 \leq i \leq \binom{n}{l}}                & \quad (\sum_{i=1}^{\binom{n}{l}} \gamma_i \; \alpha_{l,L_{i}}) & \nonumber \\
        \text{subject to}   & \quad \gamma_i \geq 0,\; 1 \leq i \leq \binom{n}{l},  & \nonumber \\
                            & \quad \sum_{i=1}^{\binom{n}{l}} \gamma_i \leq \frac{k}{l},  & \nonumber \\
                            & \quad \sum_{ \{i: I \subseteq L_i,\; 1\leq i \leq \binom{n}{l} \} } \gamma_i  \leq \frac{\binom{k-b}{l-b}}{\binom{k-1}{l-1}}, & \nonumber \\
                            &\;\begin{subarray}{l} {for\;all\;integers\; b\;such\; that\; 1\leq b \leq l,} \\ {for\; all\; subsets\; I\; with\; |I|=b} \end{subarray}.
\end{align}

\begin{lemma}
\label{pickl_optimized_coefficients_vs_the_pickl_element_algo}
The pick-$l$-element algorithm with optimized coefficients provides tighter, or at least the same, upper bound than the pick-$l$-element algorithm.
\end{lemma}

\begin{proof}
We can easily see that the following optimization problem (\ref{optimized_coefficients_without_constraint}) provides the same result as that from the pick-$l$-element algorithm:
\begin{align}
\label{optimized_coefficients_without_constraint}
        \max_{\gamma_i,\; 1 \leq i \leq \binom{n}{l} }  & \quad (\sum_{i=1}^{\binom{n}{l}} \gamma_i \; \alpha_{l,L_{i}}) & \nonumber \\
        \text{subject to}   & \quad 0 \leq \gamma_i \leq \frac{1}{{k-1 \choose l-1}},\; 1\leq i \leq \binom{n}{l}, & \nonumber \\
                            & \quad \sum_{i=1}^{\binom{n}{l}} \gamma_i \leq \frac{k}{l}. &
\end{align}
And this optimization problem (\ref{optimized_coefficients_without_constraint}) is a relaxation of the pick-$l$-element algorithm with optimized coefficients (\ref{pickl_optimized_coefficients}). Therefore, the pick-$l$-element algorithm with optimized coefficients provides tighter, or at least the same, upper bound than the pick-$l$-element algorithm.
\end{proof}

%%%%%%%%%%%%%%%%%%%%%%%%%%%%%%%%%%%%%%%%%%%%%%%%%%%%%%%%%%%%%%%%%%%%%%%%%%%%%
\begin{lemma}
\label{lemma3_various_coefficients}
The pick-$l$-element algorithm with optimized coefficients provide tighter, or at least the same, upper bounds than the pick-$a$-element algorithm with optimized coefficients when $l > a$.
\end{lemma}

\begin{proof}
From Lemma \ref{lemma:cheapupper}, we have
\begin{align*}
%\label{lemma3_optimized_coefficients_proof1}
        \alpha_{k,K}
        \leq
        \underbrace {\bigg( \frac{1}{{k-1 \choose l-1}} \bigg) \bigg( \sum_{ i=1 }^{k \choose l} \alpha_{l,L_{i}} \bigg) }_{ \substack{ \text{upper bound of $k$ elements} \\ \text{calculated with} \\ \text{the pick-$l$-element algorithm} }},\; 1\leq l \leq k,
\end{align*}
where $L_i$, $1\leq i \leq \binom{k}{l}$, are all the subsets of $K$ with cardinality $l$.
We can upper bound (\ref{pickl_optimized_coefficients}) by the following:
\begin{align}
\label{lemma3_optimized_coefficients_proof2}
        \max_{\gamma_i,\; 1 \leq j \leq \binom{n}{l}}                & \quad ( \sum_{i=1}^{\binom{n}{l}} \gamma_i \frac{1}{\binom{l-1}{a-1}} \sum_{ \{j : A_j \subset L_i, |A_j|=a \} }  \alpha_{a,A_{j}}) & \nonumber \\
        \text{subject to}   & \quad \gamma_i \geq 0,\; 1 \leq i \leq \binom{n}{l},  & \nonumber \\
                            & \quad \sum_{i=1}^{\binom{n}{l}} \gamma_i \leq \frac{k}{l},  & \nonumber \\
                            & \quad \sum_{ \{i: I \subseteq L_i,\; 1\leq i \leq \binom{n}{l} \} } \gamma_i  \leq \frac{\binom{k-b}{l-b}}{\binom{k-1}{l-1}}, & \nonumber \\
                            &\;\begin{subarray}{l} {for\;all\;integers\; b\;such\; that\; 1\leq b \leq l,} \\ {for\; all\; subsets\; I\; with\; |I|=b} \end{subarray}. &
\end{align}
(In the objective function of (\ref{lemma3_optimized_coefficients_proof2}), each $\alpha_{a,A_j},\; 1\leq j\leq \binom{n}{a}$ appears $\binom{n-a}{l-a}$ times.)
By defining $\frac{1}{\binom{l-1}{a-1}} \sum_{ \{i : A_j \subset L_i,\; 1\leq i \leq \binom{n}{l}  \} } \gamma_i$ as $\gamma_j^{'}$ and relaxing (\ref{lemma3_optimized_coefficients_proof2}),
we can obtain (\ref{lemma3_optimized_coefficients_proof3}) which is the same as the pick-$a$-element algorithm with optimized coefficients.
\begin{align}
\label{lemma3_optimized_coefficients_proof3}
        \max_{\gamma_j^{'},\; 1 \leq j \leq \binom{n}{a}}                & \quad ( \sum_{j=1}^{\binom{n}{a}} \gamma_j^{'} \; \alpha_{a,A_{j}}) & \nonumber \\
        \text{subject to}   & \quad \gamma_j^{'} \geq 0,\; 1 \leq j \leq \binom{n}{a},  & \nonumber \\
                            & \quad \sum_{j=1}^{\binom{n}{a}} \gamma_j^{'} \leq \frac{k}{a},  & \nonumber \\
                            & \quad \sum_{ \{j: I \subseteq A_j,\; 1\leq j \leq \binom{n}{a} \} } \gamma_j^{'}  \leq  \frac{ \binom{k-b}{a-b} }{ \binom{k-1}{a-1} }, & \nonumber \\
                            & \;\begin{subarray}{l} {for\;all\;integers\; b\;such\; that\; 1\leq b \leq a,} \\ {for\; all\; subsets\; I\; with\; |I|=b} \end{subarray}. &
\end{align}
In fact, the first, second, and third constraints of (\ref{lemma3_optimized_coefficients_proof3}) can be obtained from the relaxation of the constraints of (\ref{lemma3_optimized_coefficients_proof2}).
The first constraint of (\ref{lemma3_optimized_coefficients_proof3}) is trivial. The second constraint of (\ref{lemma3_optimized_coefficients_proof3}) is from the following:
\begin{align*}
%\label{lemma3_optimized_coefficients_proof5}
         \sum_{j=1}^{\binom{n}{a}} \gamma_j^{'}
        & = \sum_{j=1}^{\binom{n}{a}} \frac{1}{\binom{l-1}{a-1}} \sum_{ \{i : A_j \subset L_i,\; 1\leq i \leq \binom{n}{l} \} } \gamma_i  \nonumber \\
        & = \frac{1}{\binom{l-1}{a-1}} \binom{l}{a} \sum_{i=1}^{\binom{n}{l}} \gamma_i  \nonumber \\
        & \leq \frac{1}{\binom{l-1}{a-1}} \binom{l}{a} \frac{k}{l} \nonumber \\
        & = \frac{k}{a}
\end{align*}
The third constraint of (\ref{lemma3_optimized_coefficients_proof3}) comes from the following:
\begin{align*}
%\label{lemma3_optimized_coefficients_proof6}
        & \sum_{ \{ j: I \subseteq A_j,\; 1\leq j \leq \binom{n}{a}, \; |I|=b\} } \gamma_j^{'}\nonumber \\
        & = \sum_{ \{j: I \subseteq A_j,\; 1\leq j \leq \binom{n}{a}, \; |I|=b  \} } \frac{1}{\binom{l-1}{a-1}} \sum_{ \{i : A_j \subset L_i,\; 1\leq i \leq \binom{n}{l}  \} } \gamma_i   \nonumber \\
        & = \frac{1}{\binom{l-1}{a-1}}  \frac{ \binom{n-b}{a-b} \binom{n-a}{l-a} }{ \binom{n-b}{l-b} } \sum_{ \{i : I \subset L_i,\; 1\leq i \leq \binom{n}{l},\;|I|=b  \} } \gamma_i   \nonumber \\
        & \leq \frac{1}{\binom{l-1}{a-1}}  \frac{ \binom{n-b}{a-b} \binom{n-a}{l-a} }{ \binom{n-b}{l-b} }  \frac{\binom{k-b}{l-b}}{\binom{k-1}{l-1}}, \; 1 \leq b \leq a \nonumber \\
        & = \frac{ \binom{k-b}{a-b} }{ \binom{k-1}{a-1} }, \;1 \leq b \leq a
\end{align*}
Because (\ref{lemma3_optimized_coefficients_proof3}) is obtained from the relaxation of (\ref{lemma3_optimized_coefficients_proof2}), the optimal value of (\ref{lemma3_optimized_coefficients_proof3}) is larger or equal to the optimal value of (\ref{lemma3_optimized_coefficients_proof2}), and (\ref{lemma3_optimized_coefficients_proof3}) is nothing but the pick-$a$-element algorithm with optimized coefficients. Therefore, the pick-$l$-element algorithm provides tighter, or at least the same, upper bounds than the pick-$a$-element algorithm with optimized coefficients, when $l > a$.
\end{proof}

\section{The Sandwiching Algorithm}
\label{Sec4}
From Section \ref{Sec2}, Section \ref{Sec3} and Section \ref{Sec3-1}, we have upper bounds on $\alpha_k$ with the pick-$l$-element algorithm, $ 1\leq l \leq k$:
\begin{align*}
        \alpha_k
        \leq
        \underbrace{  \bigg( \frac{1}{{k-1 \choose l-1}} \bigg) \bigg( \sum_{ j=1 }^{k \choose l} \alpha_{l,L_{i_j}} \bigg) }_{  \substack{ \text{upper bound calculated  with}\\ \text{the pick-$l$-element algorithm}} },
\end{align*}
or the pick-$l$-element algorithm with optimized coefficients, $ 1\leq l \leq k$.
However, these algorithms do not provide the \emph{exact} value for $\alpha_k$.
In order to obtain the exact value, rather than upper bounds on $\alpha_k$, we devise a sandwiching algorithm with greatly reduced computational complexity. We remark that the convex programming methods in \cite{d2011testing} and \cite{juditsky2011verifiable} only provide upper bounds on $\alpha_k$, instead of exact values of $\alpha_k$, except when $k=1$.

The idea of our sandwiching algorithm is to maintain two bounds in computing the exact value of $\alpha_k$: an upper bound on $\alpha_k$, and a lower bound on $\alpha_k$. In algorithm execution, we constantly decrease the upper bound, and increase the lower bound. When the lower bound and upper bound meet, we immediately get a certification that the exact value of $\alpha_k$ has been reached.  There are two ways to compute the upper bounds: the `cheap' upper bound and the linear programming based upper bound. These two upper bounds are stated in Lemmas \ref{lemma:cheapupper} and \ref{lemma:lpbasedupper} respectively.

\begin{lemma}[`cheap' upper bound]
\label{lemma:cheapupper}
 Given a set $K$ with cardinality $k$, we have
\begin{align}
\label{eq:cheapbound}
        \alpha_{k,K}
        \leq
        \underbrace {\bigg( \frac{1}{{k-1 \choose l-1}} \bigg) \bigg( \sum_{ i=1 }^{k \choose l} \alpha_{l,L_{i}} \bigg) }_{ \substack{ \text{upper bound of $k$ elements} \\ \text{calculated with} \\ \text{the pick-$l$-element algorithm} }},\; 1\leq l \leq k,
\end{align}
where $\alpha_{k,K} = \frac{{\beta}_{k,K}}{1+{\beta}_{k,K}}$ and ${\beta}_{k,K}$ is defined as below, and $L_i$, $1\leq i \leq \binom{k}{l}$, are all the subsets of $K$ with cardinality $l$.

\begin{align}
\label{eq:beta_kK}
          {\beta}_{k,K} = \max_{ x \in \mathbf{R}^{m} }  & \quad \| (Hx)_{K} \|_{1}        & \nonumber \\
                          \text{subject to}              & \quad \| (Hx)_{\overline{K}} \|_{1} \leq 1 &
\end{align}
(${\beta}_{k,K}$ is defined for a given $K$ set with cardinality $k$, but ${\beta}_k$ is the maximum value over all subsets with cardinality $k$.)
\end{lemma}

\begin{proof}
This proof follows the same reasoning as in Lemma \ref{lemma3_1}. Let $L_{i}$, $1\leq i \leq \binom{n}{l}$, be the family of subsets of $K$, with cardinality $l$. It is not hard to see that each element of $K$ appears in ${ k-1 \choose l-1 }$ such subsets. In particular, we have
\begin{align*}
    K = \bigcup_{i=1}^{ k \choose l } L_i.
\end{align*}
Thus, $\forall x \in \mathbf R^m$, we can represent $\frac{ \| (Hx)_{K} \|_{1} }{ \| (Hx) \|_{1} }$ as follows.
\begin{align}
\label{Sec4_eqn1}
       \frac{ \| (Hx)_{K} \|_{1} }{ \| (Hx) \|_{1} } =
       \bigg( \frac{1}{{k-1 \choose l-1}} \bigg) \bigg( \sum_{i=1}^{k \choose l} \frac{ \| (Hx)_{ L_i  } \|_{1} }{ \| (Hx) \|_{1} } \bigg)
\end{align}
Suppose that each term of the right-hand side of (\ref{Sec4_eqn1}), $\frac{ \| (Hx)_{ L_i  } \|_{1} }{ \| (Hx) \|_{1} }$, achieves the maximum value when $x = x^{i}, \;\; i = 1,...,{k \choose l}$; and the maximum value of $\frac{\|(Hx)_{K}\|_{1}}{\|Hx\|_1}$ in (\ref{Sec4_eqn1})  is achieved when $x = x^{*}$.
Then, we have
\begin{align}
       \frac{ \| (Hx^{*})_{K} \|_{1} }{ \| (Hx^{*}) \|_{1} }
       & = \bigg( \frac{1}{{k-1 \choose l-1}} \bigg) \bigg( \sum_{i=1}^{k \choose l} \frac{ \| (Hx^{*})_{ L_i  } \|_{1} }{ \| (Hx^{*}) \|_{1} } \bigg) \nonumber \\
       & \leq \bigg( \frac{1}{{k-1 \choose l-1}} \bigg) \bigg( \sum_{i=1}^{k \choose l} \frac{ \| (Hx^{i})_{ L_i  } \|_{1} }{ \| (Hx^{i}) \|_{1} } \bigg). \nonumber
\end{align}

\end{proof}

We can also obtain the upper bound on $\alpha_{k,K}$ on a given $K$ set by solving the following optimization problem (\ref{eq:LPbasedupperboundsfor_asingleK}):
\begin{align}
\label{eq:LPbasedupperboundsfor_asingleK}
          \max               & \quad (\sum_{j=1}^{k} z_j  )   &  \nonumber \\
          \text{subject to}
          & \quad  \sum_{ \substack{ {t \in L_i}  } } z_t \leq \alpha_{l,L_i},\; \begin{subarray}{l}{\forall  L_i \subseteq K} \\ {\text{with}\; |L_i|=l, \;i=1,...,{k \choose l} }\end{subarray},&  \nonumber \\
          & \quad \; z_j \geq 0,\; j=1,2,...,k.             &
\end{align}

\begin{lemma}[linear programming based upper bound]
\label{lemma:lpbasedupper}
The optimal objective value of (\ref{eq:LPbasedupperboundsfor_asingleK}) is an upper bound on $\alpha_{k,K}$.
\end{lemma}

\begin{proof}
By the definition of ${\beta}_{k,K}$, we can write $\alpha_{k,K} (= \frac{\beta_{k,K}}{1+{\beta}_{k,K}})$ as the optimal objective value of the following optimization problem.
\begin{align}
          \max_{ x \in \mathbf{R}^{m} }  & \quad \frac{ \| (Hx)_{K} \|_{1} }{\| (Hx) \|_{1}}  &\nonumber \\
                \text{subject to}              & \quad \frac{ \| (Hx)_{L_i} \|_{1} }{\| (Hx) \|_{1}} \leq   \frac{ {\beta}_{l,L_i}}{1+{\beta}_{l,L_i}},\; i=1,...,{k \choose l}.   &
                \label{eq:withredundantconstraints}
\end{align}
This is because, by the definition of ${\beta}_{l,L_i}$, the newly added constraints $\frac{ \| (Hx)_{L_i} \|_{1} }{\| (Hx) \|_{1}} \leq   \frac{ {\beta}_{l,L_i}}{1+{\beta}_{l,L_i}},\; i=1,...,{k \choose l},$ are just redundant constraints which always hold true over $x\in \mathbf{R}^{m}$. Representing $z_t=\frac{ \| (Hx)_{\{t\}} \|_{1} }{\| (Hx) \|_{1}}$, $t=1,...,n$, we can relax (\ref{eq:withredundantconstraints}) to (\ref{eq:LPbasedupperboundsfor_asingleK}). Thus the optimal objective value of (\ref{eq:LPbasedupperboundsfor_asingleK}) is an upper bound on that of (\ref{eq:withredundantconstraints}), namely $\alpha_{k,K}$.

\end{proof}

\begin{lemma}
 The optimal objective value from (\ref{eq:LPbasedupperboundsfor_asingleK}) is no larger than that of (\ref{eq:cheapbound}).
\end{lemma}

\begin{proof}
Summing up the constraints in (\ref{eq:LPbasedupperboundsfor_asingleK})
\begin{align}
          & \sum_{ \substack{ {t \in L_i}  } } z_t \leq \underbrace{\alpha_{l,L_i}}_{\substack{ \text{Exact value of $l$ elements} } } ,\; \begin{subarray}{l} {\forall  L_i \subseteq K} \\ {\text{with}\; |L_i|=l, \;i=1,...,{k \choose l} } \end{subarray},&  \nonumber
\end{align}
we get the number in (\ref{eq:cheapbound}), since each element in $K$ appears in $\binom{k-1}{l-1}$ subsets of cardinality of $l$.
\end{proof}

%%%%%%%%%%%%%%%%%%%%%%%%%%%%%%%%%%%%%%%%%%%%%%%%%%%%%%%%%%
% Sandwiching Algorithm - pseudo code
%%%%%%%%%%%%%%%%%%%%%%%%%%%%%%%%%%%%%%%%%%%%%%%%%%%%%%%%%%
%{

\begin{algorithm}
\LinesNumbered
  \caption{Sandwiching Algorithm for computing exact value of $\alpha_k$ in Pseudo code}
  \label{sand_algo_pseudo}
  \SetAlgoLined
    {\scriptsize \tcc {Global Upper Bound: GUB } }
    {\scriptsize \tcc {Global Lower Bound: GLB } }
    {\scriptsize \tcc {Cheap Upper Bound: CUB } }
    {\scriptsize \tcc {Linear Programming based Upper Bound: LPUB } }
    {\scriptsize \tcc {Local Lower Bound: LLB } }
    \BlankLine
{\footnotesize
   \KwIn{ Sets $L_j$ with $|L_j|=l$, and $\alpha_{l,L_j}$ corresponding to $L_j,\;j=1,...,{n \choose l}$}
   %Allocate space for possible sets $K$ and their upper bounds (size: ${n \choose k} \times (k+1)$) \par
%   Assign ${n \choose k}$ possible sets $K$ into the space \par
   Compute the CUB on $\alpha_{k,K}$ for all the subsets $K$ with $|K|=k$. \par
   Sort these subsets in descending order of their CUB \par
   Initialize GLB $\leftarrow 0$ \par
   \For { $i=1$  \KwTo $n \choose k$ } {
        \uIf { GLB  $<$ the CUB of the $i$-th sorted subset $K_i$ } {
            GUB $\leftarrow$ the CUB of the $i$-th sorted subset $K_i$ \par
            LPUB $\leftarrow$ the upper bound on $\alpha_{k,K_i}$ from (\ref{eq:LPbasedupperboundsfor_asingleK}) for the $i$-th subset $K_i$ \par
            \If { GLB $<$  LPUB } {
                LLB $\leftarrow$ $\alpha_{k, K_i}$, after computing $\alpha_{k,K_i}$.   \par
                \If { LLB $>$ GLB} {
                    GLB $\leftarrow$ LLB \par
                }
            }
        }
        \Else{
            GUB $\leftarrow$ GLB \par
            break
        }
   }

   GUB$\leftarrow$ GLB; \par
   \uIf { GUB $< \frac{1}{2}$ } {
        NSC is satisfied
  }
   \Else{
        NSC is not satisfied
  }
}%
\end{algorithm}

%}%%%%%%%%%%%%%%%%%%%%%%%%%%%%%%%%%%%%%%%%%%%%%%%%%%%%%%%%%%%

%%%%%%%%%%%%%%%%%%%%%%%%%%%%%%%%%%%%%%%%%%%%%%%%%%%%%%%%%%
% Sandwiching Algorithm
%%%%%%%%%%%%%%%%%%%%%%%%%%%%%%%%%%%%%%%%%%%%%%%%%%%%%%%%%%%
%
\begin{algorithm}
  \caption{Sandwiching Algorithm for computing exact value of $\alpha_k$ in description}
    \label{sand_algo_description}
    \nlset{*}{\footnotesize \textbf{Global Upper Bound (GUB): the current upper bound $\alpha_{k}$} }\par
    \nlset{*}{\footnotesize \textbf{Global Lower Bound (GLB): the current lower bound for $\alpha_{k}$ } } \par
    \nlset{*}{\footnotesize \textbf{Cheap Upper Bound (CUB): the upper bounds obtained from (\ref{eq:cheapbound})} } \par
    \nlset{*}{\footnotesize \textbf{Linear Programming based Upper Bound (LPUB): the upper bounds obtained from (\ref{eq:LPbasedupperboundsfor_asingleK})} } \par
    \BlankLine
    \BlankLine

    \nlset{1} For a fixed number $l< k$, compute $\alpha_{l,L_j},\;j=1,...,{n \choose l}$, for all the subsets $L_j$ with $|L_j|=l$. Compute the cheap upper bounds on $\alpha_{k,K}$ for all subsets $K$ with $|K|=k$, and sort these subsets by their cheap upper bounds in descending order. \par
    \nlset{2} Initialize GLB $\leftarrow$ $0$ and the index $i$ $\leftarrow$ $1$. \par
    \nlset{3} If $i$ = ${n \choose k} + 1$, then assign GUB $\leftarrow$ GLB and go to Step 7. If the CUB of the $i$-th sorted subset is no bigger than than GLB, then assign GUB $\leftarrow$ GLB and go to Step 7. \par
    \nlset{4} Assign GUB $\leftarrow$ the CUB of the $i$-th sorted subset $K_i$, and compute the LPUB for this subset $K_i$. \par
    \nlset{5} If the LPUB of the $i$-th subset $K_i$ is bigger than GLB, then calculate the exact $\alpha_{k,K_i}$ by solving (\ref{eq:beta_kK}) and assign GLB $\leftarrow$ $\alpha_{k,K_i}$ only if $\alpha_{k,K_i}>$ GLB.  \par
    \nlset{6} Increase the index $i$ $\leftarrow$ $(i+1)$ and go to Step 3. \par
    \nlset{7} If GUB is smaller than $\frac{1}{2}$, the null space condition is satisfied. If not, the null space condition is not satisfied. \par
\end{algorithm}
%%%%%%%%%%%%%%%%%%%%%%%%%%%%%%%%%%%%%%%%%%%%%%%%%%%%%%%%%%%%

In Algorithms \ref{sand_algo_pseudo} and \ref{sand_algo_description}, we shows how we implemented the sandwiching algorithm.
The following theorem claims that Algorithms \ref{sand_algo_pseudo} and \ref{sand_algo_description} will output the exact value of $\alpha_{k}$ in a finite number of steps.
\begin{theorem}
 The global lower and upper bounds on $\alpha_{k}$  will both converge to $\alpha_{k}$ in a finite number of steps.
\end{theorem}

\begin{proof}
In the sandwiching algorithm, we first use the pick-$l$-element algorithm to calculate the values of ${\alpha}_{l,L}$ for every subset $L$ with cardinality $l$.
Then using the `cheap' upper bound (\ref{eq:cheapbound}), we calculate the upper bounds on ${\alpha}_{k,K}$ for every set $K$ with cardinality $k$.
We then sort these subsets in descending order by their upper bounds.

In algorithm execution, because of sorting, the global upper bound GUB on ${\alpha}_{k}$ never rises.
In the meanwhile, the global lower bound GLB either rises or stays unchanged in each iteration.
If the algorithm comes to an index $i$ , $1\leq i \leq \binom{n}{k}$, such that the upper bound of $\alpha_{k,K_i}$ for the $i$-th subset $K_i$ is already smaller than the global lower bound GLB,
the algorithm will make the global upper bound GUB equal to the global lower bound GLB. At that moment, we know they must both be equal to ${\alpha}_{k}$. This is because, from the descending order of the upper bounds on ${\alpha}_{k,K}$, each subset $K_j$ with $j>i$ must have an $\alpha_{k,K_j}$ that is smaller than the global lower bound GLB. In the meanwhile, as specified by the sandwiching algorithm, the global lower bound GLB is the largest among $\alpha_{k,K_j}$ with $1\leq j \leq (i-1)$. So at this point, the GLB must be the largest among $\alpha_{k,K_j}$ with $1\leq j \leq \binom{n}{k}$, namely GLB$=\alpha_k$.

If we can not find such an index $i$, the algorithm will end up calculating ${\alpha}_{k,K}$ for every set $K$ in the list. In this case, the upper and lower bound will also become equal to $\alpha_{k}$, after each ${\alpha}_{k,K}$ has been calculated.

\end{proof}

\subsection{Calculating ${\alpha}_{k,K}$ for a set $K$}
The exact value of ${\alpha}_{k,K} (= \frac{\beta_{k,K}}{1+\beta_{k,K}})$ is calculated by solving (\ref{eq:beta_kK}) for a subset $K$.
However, the objective function is not concave.
In order to solve it, we separate the $\ell_1$ norm of $(Hx)_K$ into $2^{k}$ possible cases according to the sign of each term, $+1$ or $-1$.
Hence, we can make a $\ell_1$ optimization problem into $2^k$ small linear problems.
For each possible case, we find the maximum candidate value for $\beta_{k,K}$ via the following:
\begin{align}
\label{Sec4_eqn2}
                          \max_{ x \in \mathbf{R}^{m} }  & \quad  \sum_{\substack{ i\in K \\ sign_i\in\{-1,1\}}} sign_i \times (Hx)_{\{i\}}   &  \nonumber \\
                          \text{subject to}              & \quad \| (Hx)_{\overline{K}} \|_{1} \leq 1, &
\end{align}
where $sign_i$ is for the sign of $i$-th term.
In fact, we do not need to calculate $2^{k}$ small linear problems. We only need to calculate $2^{k-1}$ problems instead of $2^{k}$,
because the result from (\ref{Sec4_eqn2}) for one possible case (e.g. 1,-1,-1,1 when $k=4$) out of $2^{k}$ cases is always equivalent to the result for its inverse case (e.g. -1,1,1,-1).
Among the $2^{k-1}$ candidates, we choose the biggest one as $\beta_{k,K}$. This strategy is also applied to solve (\ref{Sec2_eqn1}) and (\ref{Sec3_eqn1}).

\subsection{Computational Complexity}
The sandwiching algorithm consists of three major parts. The first part performs the pick-$l$-element algorithm for a fixed number $l$.
The second part is the complexity of computing the upper bounds on $\alpha_{k,K}$, and sorting the $\binom{n}{k}$ subsets $K$ by the upper bounds on $\alpha_{k,K}$ in descending order.
The third part is to exactly compute $\alpha_{k,K}$ for each subset $K$, starting from the top of the sorted list,
before the upper bound meets the lower bound in the algorithm.

The first part of the sandwiching algorithm can be finished with polynomial-time complexity, when the number $l$ is fixed. The complexity of the second part grows exponentially in $n$;
however, computing the upper bounds based on the pick-$l$-element algorithms, and ranking the upper bounds are very cheap in computation.
So when $n$ and $k$ are not big (for example, $n=40$ and $k=5$), this second step can also be finished reasonably fast. We remark that, however, when $n$ and $k$ are big, one may enumerate these branches one by one sequentially, instead of computing and ranking them in one shot (A detailed discussion of this is out of the scope of this current paper). The main complexity then comes from the third part, which depends heavily on, for how many subsets $K$ the algorithm will exactly compute $\alpha_{k,K}$,
before the upper bound and the lower bound meet. In turn, this depends on how tight the upper bound and lower bound are in algorithm execution.

In the worst case, the upper and lower bound can meet when $\binom{n}{k}$ subsets $K$ have been examined.
However, in practice, we find that, very often, the upper bounds and the lower bounds meet very quickly,
often way before the algorithm has to examine $\binom{n}{k}$ subsets.
Thus the algorithm will output the exact value of $\alpha_k$, by using much lower computational complexity than the exhaustive search method.
Intuitively, subsets with bigger upper bounds on $\alpha_{k,K}$ also tend to offer bigger exact values of $\alpha_{k,K}$.
This in turn leads to very tight lower bounds on $\alpha_k$.
As we go down the sorted list of subsets, the lower bound becomes tighter and tighter, while the upper bound also becomes tighter and tighter,
since the upper bounds were sorted in descending order.
Thus the lower and upper bounds can become equal very quickly.
In the extreme case, if both upper and lower bounds are tight at the beginning, the sandwich algorithm will be terminated at the very first step.
To analyze how quickly the upper and lower bound meet in this algorithm is a very interesting problem.

\section{Simulation Results}
\label{Sec5}

We conducted simulations using Matlab on a HP Z220 CMT workstation with Intel Corei7-3770 dual core CPU @ 3.4GHz clock speed
and 16GB DDR3 RAM, under Windows 7 OS environment.
To solve optimization problems such as (\ref{Sec2_eqn1}), (\ref{Sec3_eqn1}), (\ref{eq:beta_kK}), and (\ref{eq:LPbasedupperboundsfor_asingleK}),
we used CVX,
a package for specifying and solving convex programs \cite{cvx}.

Tables ranging from \ref{Table_Gaussian_pick1} to \ref{Table_Gaussian_102x256} are the results for Gaussian matrix cases.
Gaussian matrix $H$ was chosen randomly and simulated for various $k$ from 1 to 5.
The elements of $H$ matrix follow i.i.d. standard Gaussian distribution $\mathcal{N}(0,1)$.

Table \ref{Table_Gaussian_pick1}, \ref{Table_Gaussian_pick2} and \ref{Table_Gaussian_pick3} show upper bounds on $\alpha_k$ obtained
from the pick-$1$-element algorithm, the pick-$2$-element algorithm, and the pick-$3$-element algorithm respectively for Gaussian matrix cases.
We ran simulations on 10 different random matrices $H$ for each size and obtained median value of them.
$\alpha_{1}$ in Table \ref{Table_Gaussian_pick2} and \ref{Table_Gaussian_pick3} is from Table \ref{Table_Gaussian_pick1}
and $\alpha_{2}$ in Table \ref{Table_Gaussian_pick3} is from Table \ref{Table_Gaussian_pick2}.

Table \ref{Table_Gaussian_sandwiching} shows the exact $\alpha_k$ from the sandwiching algorithm on different sizes of $H$ matrices
and different values of $k$.
We ran simulations on one randomly chosen matrix $H$ at each size. Hence in total, we tested $4$ different $H$ matrices in this simulation
(our simulation experience shows that the performance and complexity of the sandwich algorithm concentrates for random matrices under this dimension).
The pick-$l$-element algorithm mostly used in the sandwiching algorithm is the pick-$2$-element algorithm,
except for $\alpha_2$ in all $H$ matrix cases, $\alpha_4$ in the $40 \times 20$ $H$ matrix case and $\alpha_5$ in the $40 \times 12$, $40 \times 16$, and $40 \times 20$ $H$ matrix cases.
For $\alpha_2$ in all $H$ matrix cases, the sandwiching algorithm based on the pick-$1$-element algorithm is used.
For other exceptional cases, the sandwiching algorithm based on the pick-$3$-element algorithm is used,
because of the faster running time than the sandwiching algorithm based on the pick-$2$-element algorithm.
The obtained exact $\alpha_k$ is in Table \ref{Table_Gaussian_sandwiching}
and the number of steps and running time to reach that exact $\alpha_k$ are in Table \ref{Table_Gaussian_sandwiching_step} and
Table \ref{Table_Gaussian_sandwiching_time} respectively.
We cited the results from \cite{d2011testing} and \cite{juditsky2011verifiable} in Table \ref{Table_Gaussian_cite}  for easy comparison with our results.
The exact values $\alpha_k$ from our algorithm clearly improve on the upper and lower bounds from \cite{d2011testing} and \cite{juditsky2011verifiable}.
We added one more column in Table \ref{Table_Gaussian_cite} for maximum $k$ satisfying $\alpha_k < \frac{1}{2}$ based on their results.
In the $40 \times 12$ and $40 \times 16$ $H$ matrix cases, we have bigger $k$ than \cite{d2011testing} and \cite{juditsky2011verifiable}.

Table \ref{Table_Gaussian_sandwiching_step} shows the number of running steps to get the exact $\alpha_k$ in Table \ref{Table_Gaussian_sandwiching},
using our sandwiching algorithm.
As shown in Table \ref{Table_Gaussian_sandwiching_step}, we can reduce running steps considerably in reaching the exact $\alpha_k$,
compared with the exhaustive search method.
When $k=3$, for the $40 \times 16$ $H$ matrix case, the number of running steps was reduced to about $\frac{1}{700}$ of the steps in the exhaustive search method.
The running steps for $k=4$ and the same $H$ matrix are reduced to about $\frac{1}{40}$ of the steps in the exhaustive search method.
In $k=5$ case, the reduction rate became $\frac{1}{5}$ on the same $H$ matrix.
We think that this is because when $k$ is big, the gap between the upper bound on $\alpha_k$ from the pick-$2$-element algorithm, and
the lower bound becomes big, thus the number of running steps is increased.
(With the sandwiching algorithm based on the pick-$3$-element algorithm in $k=5$ and the $40 \times 16$ $H$ matrix case,
the reduction rate became $\frac{1}{4400}$.)

Table \ref{Table_Gaussian_sandwiching_time} lists the actual running time of the sandwiching algorithm (mostly based on the pick-$2$-element algorithm).
Except for $k=2$, the pick-$1$-element algorithm is used as the steps in the sandwiching algorithm.
For $k=4$ in the $40 \times 20$ $H$ matrix case, and $k=5$ in the $40 \times 12$, $40 \times 16$ and $40 \times 12$ cases,
the pick-$3$-element algorithm is used in the sandwiching algorithm.
For $k=5$ in the $40 \times 20$ $H$ matrix, our sandwiching algorithm finds the exact value using only $\frac{1}{170}$ of the time used by the exhaustive search method:
the sandwiching algorithm takes around $2.2$ hours, while the exhaustive search method will take around 16 days to find the exact value of $\alpha_k$.

Table \ref{Table_Gaussian_exhaustive} shows the estimated running time of the exhaustive search method.
In order to estimate the running time, we measured the running time to obtain $\alpha_{k,K}$ for $100$ randomly chosen subsets $K$ with $|K|=k$,
and calculated the average time spent per subset.
We multiplied the time per subset with the number of subsets in the exhaustive search method to calculate the overall running time of the exhaustive search method.
For $k=1$ case, we put the actual operation time from Table \ref{Table_Gaussian_sandwiching_time}.

Tables ranging from \ref{Table_Fourier_pick1} to \ref{Table_Fourier_cite} are the results for Fourier matrix cases
and Tables ranging from \ref{Table_Bernoulli_pick1} to \ref{Table_Bernoulli_cite} are the results for Bernoulli matrix cases.
For Fourier matrix cases and Bernoulli matrix cases, we used $A$ matrix in simulations instead of its null space matrix.
$A$ matrix was chosen randomly and simulated for various $k$ from 1 to 5.

Table \ref{Table_Fourier_pick1}, Table \ref{Table_Fourier_pick2}, and Table \ref{Table_Fourier_pick3} show upper bounds on $\alpha_k$ for Fourier matrix cases.
We ran simulations on 10 different random Fourier matrices for each size and obtained median value of them.
$\alpha_{1}$ in Table \ref{Table_Fourier_pick2} and \ref{Table_Fourier_pick3} is from Table \ref{Table_Fourier_pick1}
and $\alpha_{2}$ in Table \ref{Table_Fourier_pick3} is from Table \ref{Table_Fourier_pick2}.

Table \ref{Table_Fourier_sandwiching} shows the exact $\alpha_k$ from the sandwiching algorithm on different sizes of Fourier matrices $A$
and different values of $k$.
We ran simulations on one randomly chosen Fourier matrix $A$ at each size. Hence in total, we tested $4$ different Fourier matrices in this simulation.
The obtained exact $\alpha_k$ via our sandwiching algorithm is in Table \ref{Table_Fourier_sandwiching}
and the number of steps and running time to reach that exact $\alpha_k$ are in Table \ref{Table_Fourier_sandwiching_step} and
Table \ref{Table_Fourier_sandwiching_time} respectively.
The exact values $\alpha_k$ from our algorithm clearly improve on the upper and lower bounds from \cite{d2011testing} and \cite{juditsky2011verifiable} for Fourier matrix cases as well.
For example, when our result for $\alpha_k$ is compared to the results from \cite{d2011testing} and \cite{juditsky2011verifiable} in $20 \times 40$ Fourier matrix, we obtained 0.67 for exact $\alpha_5$, while both \cite{d2011testing} and \cite{juditsky2011verifiable} provides 0.98 as their upper bounds on $\alpha_5$.)

In Table \ref{Table_Fourier_sandwiching}, the pick-$l$-element algorithm mostly used in the sandwiching algorithm is the pick-$2$-element algorithm,
except for $\alpha_2$, and $\alpha_5$ in the $20 \times 40$ and $32 \times 40$ matrix cases.
For $\alpha_2$, the sandwiching algorithm based on the pick-$1$-element algorithm is used.
For $\alpha_5$ in the $20 \times 40$ and $32 \times 40$ matrix cases, the sandwiching algorithm based on the pick-$3$-element algorithm is used,
because of the faster running time than the sandwiching algorithm based on the pick-$2$-element algorithm. In Table \ref{Table_Fourier_sandwiching_step}, when $k=2$,
some results from the sandwiching algorithm based on the pick-$1$-element algorithm reached the maximum operation steps, namely $\binom{n}{k}$ steps.
This is because in those Fourier matrices, the upper bounds obtained from (\ref{eq:cheapbound}) and (\ref{eq:LPbasedupperboundsfor_asingleK}) were too weak to satisfy conditions in the program to stop the simulation in the middle of the operation early.

Table \ref{Table_Bernoulli_pick1}, Table \ref{Table_Bernoulli_pick2}, and Table \ref{Table_Bernoulli_pick3} show upper bounds on $\alpha_k$ for Bernoulli matrix cases.
We ran simulations on 10 different random Bernoulli matrices for each size and obtained median value of them.
$\alpha_{1}$ in Table \ref{Table_Bernoulli_pick2} and \ref{Table_Bernoulli_pick3} is from Table \ref{Table_Bernoulli_pick1}
and $\alpha_{2}$ in Table \ref{Table_Bernoulli_pick3} is from Table \ref{Table_Bernoulli_pick2}.

Table \ref{Table_Bernoulli_sandwiching} shows the exact $\alpha_k$ from the sandwiching algorithm on different sizes of Bernoulli matrices $A$
and different values of $k$.
We ran simulations on one randomly chosen Bernoulli matrix $A$ at each size. Hence in total, we tested $4$ different Bernoulli matrices in this simulation.
The obtained exact $\alpha_k$ via our sandwiching algorithm is in Table \ref{Table_Bernoulli_sandwiching}
and the number of steps and running time to reach that exact $\alpha_k$ are in Table \ref{Table_Bernoulli_sandwiching_step} and
Table \ref{Table_Bernoulli_sandwiching_time} respectively.

In Table \ref{Table_Bernoulli_sandwiching}, in order to obtain $\alpha_2$,
the pick-$1$-element algorithm is used in the sandwiching algorithm (the sandwiching algorithm based on the pick-$1$-element algorithm).
For $\alpha_5$ in the $20 \times 40$, $24 \times 40$, and $28 \times 40$,
the pick-$3$-element algorithm is used in the sandwiching algorithm (the sandwiching algorithm based on the pick-$3$-element algorithm).
When it comes to comparing the results from our sandwiching algorithm with the results from \cite{d2011testing} and \cite{juditsky2011verifiable},
we have bigger recoverable $k$ in the $24 \times 40$, $28 \times 40$, and $32 \times 40$ Bernoulli matrix cases.

Figures \ref{hist_alpha5}, \ref{hist_steps}, and \ref{hist_time} show histograms of the sandwiching algorithm conducted on
100 examples $40 \times 20$ Gaussian matrices. The sandwiching algorithm based on the pick-$3$-element algorithm is used in our simulations. The median value of $\alpha_5$, number of steps and operation time of the sandwiching algorithm in this 100 trials are respectively 0.73, 1400 steps, and 95.36 minutes. The data of 10 samples out of 100 trials are in Table \ref{Table_Gaussian_sandwiching_10trial}. We remark that, if one attempts to use exhaustive search to get the exact $\alpha_k$ for these $100$ matrices, it would take around $4$ years on our machine.

\begin{figure}[!ht]
    \caption{Histogram of the exact $\alpha_5$ from the sandwiching algorithm ($40 \times 20$ Gaussian matrix) }
    \label{hist_alpha5}
    \centering
    \includegraphics[scale=0.5]{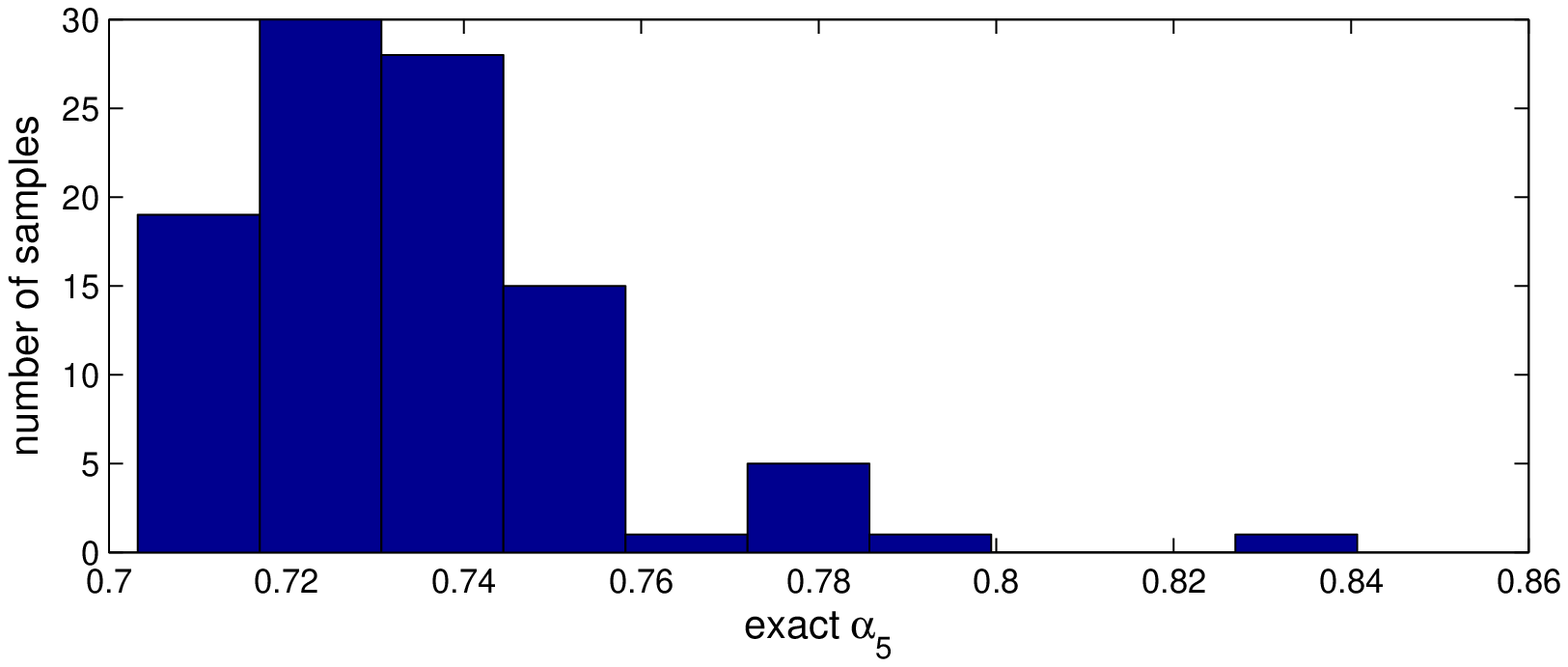}
\end{figure}
\begin{figure}[!ht]
    \caption{Histogram of the number of steps in the sandwiching algorithm ($40 \times 20$ Gaussian matrix) }
    \label{hist_steps}
    \centering
    \includegraphics[scale=0.5]{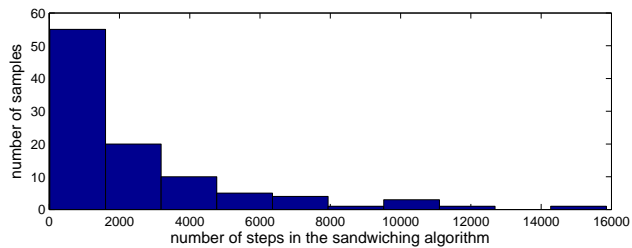}
\end{figure}
\begin{figure}[!ht]
    \caption{Histogram of operation time in the sandwiching algorithm ($40 \times 20$ Gaussian matrix) }
    \label{hist_time}
    \centering
    \includegraphics[scale=0.5]{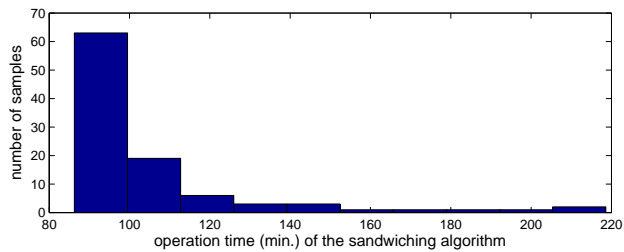}
\end{figure}

Figure \ref{upper_lower_graph} shows how fast the upper bound and lower bound are approaching each other in the sandwiching algorithm (based on the pick-$2$-element algorithm), for $k=5$ and $40 \times 20$ $H$ Gaussian matrix case. We can see that, the sandwiching algorithm offers a good tradeoff between result accuracy and computation complexity, if we ever want to terminate the algorithm early.

\begin{figure}[!ht]
    \caption{Global Upper Bound (GUB) and Global Lower Bound (GLB) in the sandwiching algorithm based on the pick-$2$-element algorithm ($\alpha_5$ in the $40 \times 20$ $H$ Gaussian matrix case) }
    \label{upper_lower_graph}
    \centering
    \includegraphics[scale=0.5]{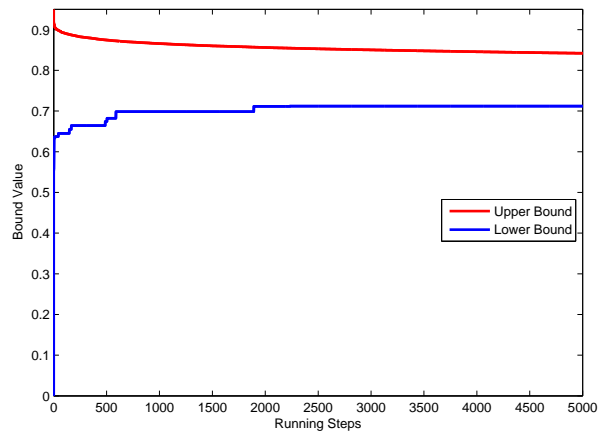}
\end{figure}

Figure \ref{102x256} is the graph for the upper bound on $\alpha_k$ versus $k$ in $102 \times 256$ $A$ Gaussian matrix, ($A \in \mathbf R^{102 \times 256}$).
We obtained $k = 5$ ($\alpha_5$ = 0.49) from the pick-$2$-element algorithm as maximum $k$ such that $\alpha_k < 1/2$ in $102 \times 256$ Gaussian matrix,
while \cite{juditsky2011verifiable} provided 4 for recoverable sparsity in a Gaussian matrix of the same dimension.
($\alpha_1$ in the pick-$2$-element algorithm comes from $\alpha_1$ in the pick-$1$-element algorithm.)
The data are in Table \ref{Table_Gaussian_102x256}.

\begin{figure}[!ht]
    \caption{Upper Bounds on $\alpha_k$ from the pick-$1$-element algorithm and the pick-$2$-element algorithm ($102 \times 256$ $A$ Gaussian matrix) }
    \label{102x256}
    \centering
    \includegraphics[scale=0.5]{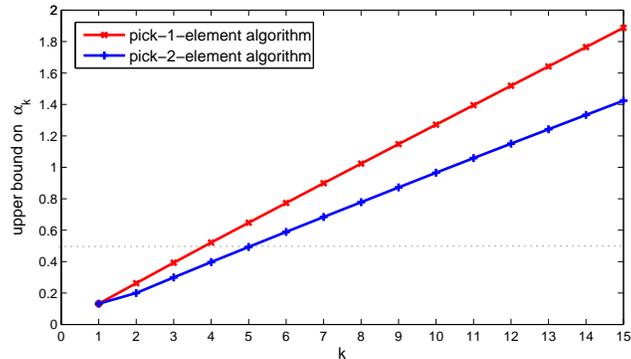}
\end{figure}

\section{Conclusion}
\label{Sec6}

In this paper, we proposed new algorithms to verify the null space conditions.
We first proposed a series of new polynomial-time algorithms to compute upper bounds on $\alpha_k$.
Based on these new polynomial-time algorithms, we further designed a new sandwiching algorithm, to compute the \emph{exact} $\alpha_k$ with greatly reduced complexity.

The future work for verifying the null space conditions includes designing efficient algorithms to reduce the operation time even more.
It is also interesting to extend the framework to the nonlinear measurement setting \cite{xu2011sparse}.

\bibliographystyle{IEEEbib}
\bibliography{refs}

\newpage
%%%%%%%%%%%%%%%%%%%%%%%%%%%%%%%%%%%%%%%%%%%%%%%%%
%\FloatBarrier
\begin{table*}[!ht]
\centering

\begin{threeparttable}[t]
\caption{ Upper bounds from the pick-$1$-element algorithm (Gaussian Matrix)}

\label{Table_Gaussian_pick1}
\setlength{\tabcolsep}{15pt}

{\small
    \begin{tabular}{cccccccc}
        \multicolumn{8}{r} {\scriptsize (Rounded off to the nearest hundredth) }\\
        \hline \hline
        H matrix(n $\times$ m) &{$\rho$}\tnote{a}  & $\alpha_1$    & $\alpha_2$    & $\alpha_3$   & $\alpha_4$   & $\alpha_5$ & $k$\tnote{b}  \\
        \hline
        40 $\times$ 20         &0.5        &0.27      &0.54      &0.79     &1.03     &1.27   & 1  \\
        40 $\times$ 16         &0.6        &0.23      &0.44      &0.65     &0.86     &1.06   & 2  \\
        40 $\times$ 12         &0.7        &0.19      &0.36      &0.53     &0.70     &0.86   & 2  \\
        40 $\times$ 8          &0.8        &0.15      &0.29      &0.43     &0.56     &0.69   & 3  \\
        \hline
    \end{tabular}
}
\begin{tablenotes}
%    \item {\footnotesize Rounded off to the nearest hundredth }
    \item [a]{\footnotesize $\rho = (n-m)/n$ }
    \item [b]{\footnotesize Maximum $k$ s.t. $\alpha_k < 1/2$ }
\end{tablenotes}

\end{threeparttable}
\end{table*}

%%%%%%%%%%%%%%%%%%%%%%%%%%%%%%%%%%%%%%%%%%%%%%%%%
%\FloatBarrier
\begin{table*}[ht]
\centering

\begin{threeparttable}
\caption{ Upper bounds from the pick-$2$-element algorithm (Gaussian Matrix)}

\label{Table_Gaussian_pick2}
\setlength{\tabcolsep}{15pt}

{\small
    \begin{tabular}{cccccccc}
        \multicolumn{8}{r} {\scriptsize (Rounded off to the nearest hundredth) }\\
        \hline \hline
        H matrix(n $\times$ m) &$\rho$\tnote{a} & $\alpha_1$    & $\alpha_2$    & $\alpha_3$   & $\alpha_4$   & $\alpha_5$ & $k$\tnote{b} \\
        \hline
        40 $\times$ 20         &0.5        &0.27      &0.46      &0.65     &0.83     &1.02   &  2 \\
        40 $\times$ 16         &0.6        &0.23      &0.37      &0.53     &0.69     &0.85   &  2 \\
        40 $\times$ 12         &0.7        &0.19      &0.32      &0.46     &0.60     &0.73   &  3 \\
        40 $\times$ 8          &0.8        &0.15      &0.25      &0.37     &0.48     &0.59   &  4 \\
        \hline
    \end{tabular}
}
\begin{tablenotes}
%    \item {\footnotesize Rounded off to the nearest hundredth }
    \item [a]{\footnotesize $\rho = (n-m)/n$ }
    \item [b]{\footnotesize Maximum $k$ s.t. $\alpha_k < 1/2$ }
\end{tablenotes}
\end{threeparttable}
\end{table*}

%%%%%%%%%%%%%%%%%%%%%%%%%%%%%%%%%%%%%%%%%%%%%%%%%
%\FloatBarrier
\begin{table*}[ht]
\centering

\begin{threeparttable}
\caption{ Upper bounds from the pick-$3$-element algorithm (Gaussian Matrix)}

\label{Table_Gaussian_pick3}
\setlength{\tabcolsep}{15pt}

{\small
    \begin{tabular}{cccccccc}
        \multicolumn{8}{r} {\scriptsize (Rounded off to the nearest hundredth) }\\
        \hline \hline
        H matrix(n $\times$ m) &$\rho$\tnote{a} & $\alpha_1$    & $\alpha_2$    & $\alpha_3$   & $\alpha_4$   & $\alpha_5$ & $k$\tnote{b} \\
        \hline
        40 $\times$ 20         &0.5        &0.27      &0.46      &0.55     &0.72     &0.88   &  2 \\
        40 $\times$ 16         &0.6        &0.23      &0.37      &0.47     &0.61     &0.74   &  3 \\
        40 $\times$ 12         &0.7        &0.19      &0.32      &0.41     &0.54     &0.65   &  3 \\
        40 $\times$ 8          &0.8        &0.15      &0.25      &0.33     &0.43     &0.52   &  4 \\
        \hline
    \end{tabular}
}
    \begin{tablenotes}
%        \item {\footnotesize Rounded off to the nearest hundredth }
        \item [a]{\footnotesize $\rho = (n-m)/n$ }
        \item [b]{\footnotesize Maximum $k$ s.t. $\alpha_k < 1/2$ }
    \end{tablenotes}

\end{threeparttable}
\end{table*}

%%%%%%%%%%%%%%%%%%%%%%%%%%%%%%%%%%%%%%%%%%%%%%%%%
%\FloatBarrier
\begin{table*}[ht]
\centering

\begin{threeparttable}
\caption{ Exact $\alpha_k$ from the sandwiching algorithm (Gaussian Matrix)}

\label{Table_Gaussian_sandwiching}
\setlength{\tabcolsep}{15pt}

{\small
    \begin{tabular}{cccccccc}
        \multicolumn{8}{r} {\scriptsize (Rounded off to the nearest hundredth) }\\
        \hline \hline
        H matrix(n $\times$ m) &$\rho$\tnote{a} & $\alpha_1$    & $\alpha_2$\tnote{c}    & $\alpha_3$   & $\alpha_4$   & $\alpha_5$ & $k$\tnote{b} \\
        \hline
        40 $\times$ 20         &0.5        & 0.27     & 0.42     & 0.54    & 0.63\tnote{d}      & 0.71\tnote{d}       &  2 \\
        40 $\times$ 16         &0.6        & 0.22     & 0.38     & 0.46    & 0.55               & 0.63\tnote{d}       &  3 \\
        40 $\times$ 12         &0.7        & 0.17     & 0.27     & 0.36    & 0.44               & 0.52\tnote{d}       &  4 \\
        40 $\times$ 8          &0.8        & 0.15     & 0.27     & 0.36    & 0.42               & 0.50                &  4 \\
        \hline
    \end{tabular}
}
    \begin{tablenotes}
%        \item {\footnotesize Rounded off to the nearest hundredth }
        \item[a] {\footnotesize $\rho = (n-m)/n$ }
        \item[b] {\footnotesize Maximum $k$ s.t. $\alpha_k < 1/2$ }
        \item[c] {\footnotesize Obtained from the sandwiching algorithm based on the pick-$1$-element algorithm }
        \item[d] {\footnotesize Obtained from the sandwiching algorithm based on the pick-$3$-element algorithm }
        \\
    \end{tablenotes}

\end{threeparttable}
\end{table*}

%%%%%%%%%%%%%%%%%%%%%%%%%%%%%%%%%%%%%%%%%%%%%%%%%
%\FloatBarrier
\begin{table*}[h]
\centering

\begin{threeparttable}
\caption{ Upper and lower bounds when $n=40$ from \cite{d2011testing} and \cite{juditsky2011verifiable} (Gaussian Matrix)}

\label{Table_Gaussian_cite}
\setlength{\tabcolsep}{15pt}

{\small
    \begin{tabular}{cccccccc}
        \multicolumn{7}{r} {\scriptsize }\\
        \hline \hline
        Relaxation &$\rho$     &$\alpha_1$     &$\alpha_2$    & $\alpha_3$   & $\alpha_4$   & $\alpha_5$ & $k$\tnote{c} \\
        \hline
        LP\tnote{a}   &0.5        & 0.27     & 0.49    & 0.67    & 0.83    & 0.97  &  2 \\
        SDP\tnote{b} &0.5        & 0.27     & 0.49    & 0.65    & 0.81    & 0.94  &  2 \\
        SDP low.   &0.5        & 0.27     & 0.31    & 0.33    & 0.32    & 0.35  &  2 \\
        \hline
        LP         &0.6        & 0.22     & 0.41    & 0.57    & 0.72    & 0.84  &  2 \\
        SDP        &0.6        & 0.22     & 0.41    & 0.56    & 0.70    & 0.82  &  2 \\
        SDP low.   &0.6        & 0.22     & 0.29    & 0.31    & 0.32    & 0.36  &  2 \\
        \hline
        LP         &0.7        & 0.20     & 0.34    & 0.47    & 0.60    & 0.71  &  3 \\
        SDP        &0.7        & 0.20     & 0.34    & 0.46    & 0.59    & 0.70  &  3 \\
        SDP low.   &0.7        & 0.20     & 0.27    & 0.31    & 0.35    & 0.38  &  3 \\
        \hline
        LP         &0.8        & 0.15     & 0.26    & 0.37    & 0.48    & 0.58  &  4 \\
        SDP        &0.8        & 0.15     & 0.26    & 0.37    & 0.48    & 0.58  &  4 \\
        SDP low.   &0.8        & 0.15     & 0.23    & 0.28    & 0.33    & 0.38  &  4 \\
        \hline
    \end{tabular}
}
    \begin{tablenotes}
        \item[a] {\footnotesize Linear Programming }
        \item[b] {\footnotesize Semidefinite Programming }
        \item[c] {\footnotesize Maximum $k$ s.t. $\alpha_k < 1/2$ }
    \end{tablenotes}

\end{threeparttable}
\end{table*}

%%%%%%%%%%%%%%%%%%%%%%%%%%%%%%%%%%%%%%%%%%%%%%%%%
\FloatBarrier
\begin{table*}[ht]
\centering

\begin{threeparttable}
\caption{ Running steps in the sandwiching algorithm (Gaussian Matrix)}

\label{Table_Gaussian_sandwiching_step}
\setlength{\tabcolsep}{15pt}

{\small
    \begin{tabular}{ccccccc}
        \multicolumn{7}{r} {\scriptsize }\\
        \hline \hline
        H matrix(n $\times$ m) &$\rho$\tnote{a}     & $k=1$\tnote{b}  & $k=2$\tnote{c}     & $k=3$      & $k=4$         & $k=5$     \\
        \hline
        Exhaustive Search      &-              & -          & 780     & 9880     & 91390       & 658008  \\
        40 $\times$ 20         &0.5            & -          & 194     & 77       & 19\tnote{d} & 3897\tnote{d}  \\
        40 $\times$ 16         &0.6            & -          & 43      & 14       & 2362        & 148\tnote{d}  \\
        40 $\times$ 12         &0.7            & -          & 179     & 25       & 2141        & 78\tnote{d}   \\
        40 $\times$ 8          &0.8            & -          & 5       & 3        & 87          & 702     \\
        \hline
    \end{tabular}
}
    \begin{tablenotes}
        \item[a]{\footnotesize $\rho = (n-m)/n$ }
        \item[b]{\footnotesize Sandwiching algorithm is not applied }
        \item[c]{\footnotesize Obtained from the sandwiching algorithm based on the pick-$1$-element algorithm }
        \item[d]{\footnotesize Obtained from the sandwiching algorithm based on the pick-$3$-element algorithm }
    \end{tablenotes}

\end{threeparttable}

\end{table*}

%\FloatBarrier
\begin{table*}[!ht]
\centering

\begin{threeparttable}
\caption{ Running time of the sandwiching algorithm (Gaussian Matrix)}

\label{Table_Gaussian_sandwiching_time}
\setlength{\tabcolsep}{15pt}
{\small
    \begin{tabular}{ccccccc}
        \multicolumn{7}{r} {\scriptsize (Unit: minute)}\\
        \hline \hline
        H matrix(n $\times$ m) &$\rho$\tnote{a} & $k=1$     & $k=2$\tnote{b}        & $k=3$      & $k=4$         & $k=5$     \\
        \hline
        40 $\times$ 20         &0.5        & 0.10    & 2.23       & 4.20        & 89.54\tnote{c}& 133.93\tnote{c} \\
        40 $\times$ 16         &0.6        & 0.12    & 0.59       & 3.63        & 14.54         & 92.13\tnote{c}\\
        40 $\times$ 12         &0.7        & 0.11    & 2.05       & 3.76        & 16.15         & 92.04\tnote{c} \\
        40 $\times$ 8          &0.8        & 0.10    & 0.17       & 3.52        & 4.12          & 8.17   \\
        \hline
    \end{tabular}
}
    \begin{tablenotes}
%        \item {\footnotesize Rounded off to the nearest hundredth (Unit: minute)}
        \item[a]{\footnotesize $\rho = (n-m)/n$ }
        \item[b]{\footnotesize Obtained from the sandwiching algorithm based on the pick-$1$-element algorithm }
        \item[c]{\footnotesize Obtained from the sandwiching algorithm based on the pick-$3$-element algorithm }
    \end{tablenotes}

\end{threeparttable}
\end{table*}

%\FloatBarrier
\begin{table*}[!ht]
\centering

\begin{threeparttable}
\caption{Estimated running time\tnote{a} of the exhaustive search method (Gaussian Matrix)}

\label{Table_Gaussian_exhaustive}
\setlength{\tabcolsep}{15pt}

{\small
    \begin{tabular}{ccccccc}
        \multicolumn{7}{r} {\scriptsize (Unit: minute)}\\      \hline \hline %43
        H matrix(n $\times$ m) &$\rho$     & $k=1$\tnote{b}     & $k=2$       & $k=3$      & $k=4$         & $k=5$     \\
        \hline
        40 $\times$ 20         &0.5        & 0.10    & 3.39       & 86.93        & 1585        & 2.3047e4 \\
        40 $\times$ 16         &0.6        & 0.12    & 3.29       & 86.13        & 1610        & 2.3699e4 \\
        40 $\times$ 12         &0.7        & 0.11    & 3.38       & 86.33        & 1611        & 2.3247e4 \\
        40 $\times$ 8          &0.8        & 0.10    & 3.40       & 85.45        & 1609        & 2.3318e4   \\
        \hline
    \end{tabular}
}
    \begin{tablenotes}
%        \item {\footnotesize Unit: minute}
        \item[a] {\footnotesize Estimated running time = running time per step $\times$ total number of steps in exhaustive search method}
        \item[b] {\footnotesize From Table \ref{Table_Gaussian_sandwiching_time} }
    \end{tablenotes}

\end{threeparttable}

\end{table*}

%\FloatBarrier
\begin{table*}[!ht]
\centering

\begin{threeparttable}
\caption{ The sandwiching algorithm\tnote{a} - 10 samples (40 $\times$ 20 Gaussian Matrix) }

\label{Table_Gaussian_sandwiching_10trial}
\setlength{\tabcolsep}{7pt}
{\small
    \begin{tabular}{cccccccccccc}
        \multicolumn{12}{r} {\scriptsize (Rounded off to the nearest hundredth)}\\
        \hline \hline
        trial                   &1      &2      &3      &4      &5      &6      &7      &8      &9      &10     &median\tnote{b}   \\
        \hline
        $\alpha_5$              &0.75	&0.73	&0.73	&0.79	&0.72	&0.72	&0.72	&0.74	&0.74	&0.76	&0.74    \\
        number of steps         &169	&1582	&1930	&10	    &807	&3549	&1033	&767	&464	&454	&787     \\
        operation time(min.)    &88.99	&101.37	&104.51	&87.18	&90.54	&104.06	&92.09	&90.20	&88.34	&89.65	&90.37	 \\
        \hline
    \end{tabular}
}
    \begin{tablenotes}
        \item[a]{\footnotesize The sandwiching algorithm based on the pick-$3$-element algorithm }
        \item[b]{\footnotesize Median value of 10 samples in the table. }
    \end{tablenotes}

\end{threeparttable}
\end{table*}

%\FloatBarrier
\begin{table*}[!ht]
\centering

\begin{threeparttable}
\caption{ Upper bounds on $\alpha_k$ ($102 \times 256$ Gaussian Matrix)\tnote{a} }

\label{Table_Gaussian_102x256}
\setlength{\tabcolsep}{7pt}
{\small
    \begin{tabular}{cccccccccccccccc}
        \multicolumn{16}{r} {\scriptsize (Rounded off to the nearest hundredth)}\\
        \hline \hline
        k                                   &1      &2      &3      &4      &5      &6      &7      &8      &9      &10     &11    &12    &13    &14    &15   \\
        \hline
        $\alpha_k$ from pick-$1$\tnote{b}   &0.13   &0.26   &0.39   &0.52   &0.65   &0.77   &0.90   &1.02   &1.15   &1.27   &1.40  &1.52  &1.64  &1.76  &1.89  \\
        $\alpha_k$ from pick-$2$\tnote{c}   &0.13   &0.20   &0.30   &0.40   &0.49   &0.59   &0.68   &0.78   &0.87   &0.97   &1.06  &1.15  &1.24  &1.33  &1.42  \\
        \hline
    \end{tabular}
}
    \begin{tablenotes}
        \item[a]{\footnotesize $102 \times 256$ $A$ Gaussian matrix matrix ($256 \times 154$ $H$ Gaussian matrix) }
        \item[b]{\footnotesize The pick-$1$-element algorithm}
        \item[c]{\footnotesize The pick-$2$-element algorithm}

    \end{tablenotes}

\end{threeparttable}
\end{table*}

\FloatBarrier

%%%%%%%%%%%%%%%%%%%%%%%%%%%%%%%%%%%%%%%%%%%%%%%%%%%%%%%%%%%%%%%%% Fourier Matrix
\newpage
%%%%%%%%%%%%%%%%%%%%%%%%%%%%%%%%%%%%%%%%%%%%%%%%%
%\FloatBarrier
\begin{table*}[!ht]
\centering

\begin{threeparttable}[t]
\caption{ Upper bounds from the pick-$1$-element algorithm (Fourier Matrix)}

\label{Table_Fourier_pick1}
\setlength{\tabcolsep}{15pt}

{\small
    \begin{tabular}{cccccccc}
        \multicolumn{8}{r} {\scriptsize (Rounded off to the nearest hundredth) }\\
        \hline \hline
        A matrix((n-m) $\times$ n)\tnote{a} &{$\rho$}\tnote{b}  & $\alpha_1$    & $\alpha_2$    & $\alpha_3$   & $\alpha_4$   & $\alpha_5$ & $k$\tnote{c}  \\
        \hline
        20 $\times$ 40         &0.5        &0.20      &0.41      &0.61     &0.81     &1.01   & 2  \\
        24 $\times$ 40         &0.6        &0.15      &0.31      &0.46     &0.61     &0.77   & 3  \\
        28 $\times$ 40         &0.7        &0.13      &0.26      &0.39     &0.52     &0.64   & 3  \\
        32 $\times$ 40         &0.8        &0.10      &0.19      &0.29     &0.38     &0.48   & 5  \\
        \hline
    \end{tabular}
}
\begin{tablenotes}
%    \item {\footnotesize Rounded off to the nearest hundredth }
    \item [a]{\footnotesize $Az = 0$ }
    \item [b]{\footnotesize $\rho = (n-m)/n$ }
    \item [c]{\footnotesize Maximum $k$ s.t. $\alpha_k < 1/2$ }
\end{tablenotes}

\end{threeparttable}
\end{table*}

%%%%%%%%%%%%%%%%%%%%%%%%%%%%%%%%%%%%%%%%%%%%%%%%%
%\FloatBarrier
\begin{table*}[ht]
\centering

\begin{threeparttable}
\caption{ Upper bounds from the pick-$2$-element algorithm (Fourier Matrix)}

\label{Table_Fourier_pick2}
\setlength{\tabcolsep}{15pt}

{\small
    \begin{tabular}{cccccccc}
        \multicolumn{8}{r} {\scriptsize (Rounded off to the nearest hundredth) }\\
        \hline \hline
        A matrix((n-m) $\times$ n)\tnote{a} &$\rho$\tnote{b} & $\alpha_1$    & $\alpha_2$    & $\alpha_3$   & $\alpha_4$   & $\alpha_5$ & $k$\tnote{c} \\
        \hline
        20 $\times$ 40          &0.5        &0.20      &0.34      &0.52     &0.69     &0.86   & 2  \\
        24 $\times$ 40          &0.6        &0.15      &0.30      &0.46     &0.61     &0.76   & 3  \\
        28 $\times$ 40          &0.7        &0.13      &0.23      &0.35     &0.46     &0.58   & 4  \\
        32 $\times$ 40          &0.8        &0.10      &0.18      &0.26     &0.35     &0.44   & 5  \\
        \hline
    \end{tabular}
}
\begin{tablenotes}
%    \item {\footnotesize Rounded off to the nearest hundredth }
    \item [a]{\footnotesize $Az = 0$ }
    \item [b]{\footnotesize $\rho = (n-m)/n$ }
    \item [c]{\footnotesize Maximum $k$ s.t. $\alpha_k < 1/2$ }
\end{tablenotes}
\end{threeparttable}
\end{table*}

%%%%%%%%%%%%%%%%%%%%%%%%%%%%%%%%%%%%%%%%%%%%%%%%%
%\FloatBarrier
\begin{table*}[ht]
\centering

\begin{threeparttable}
\caption{ Upper bounds from the pick-$3$-element algorithm (Fourier Matrix)}

\label{Table_Fourier_pick3}
\setlength{\tabcolsep}{15pt}

{\small
    \begin{tabular}{cccccccc}
        \multicolumn{8}{r} {\scriptsize (Rounded off to the nearest hundredth) }\\
        \hline \hline
        A matrix((n-m) $\times$ n)\tnote{a} &$\rho$\tnote{b} & $\alpha_1$    & $\alpha_2$    & $\alpha_3$   & $\alpha_4$   & $\alpha_5$ & $k$\tnote{c} \\
        \hline
        20$\times$ 40          &0.5        &0.20      &0.34      &0.47     &0.62     &0.78   & 3 \\
        24 $\times$ 40         &0.6        &0.15      &0.30      &0.36     &0.49     &0.61   & 4 \\
        28 $\times$ 40         &0.7        &0.13      &0.23      &0.32     &0.42     &0.52   & 4 \\
        32 $\times$ 40         &0.8        &0.10      &0.18      &0.26     &0.34     &0.43   & 5 \\
        \hline
    \end{tabular}
}
    \begin{tablenotes}
%        \item {\footnotesize Rounded off to the nearest hundredth }
        \item [a]{\footnotesize $Az = 0$ }
        \item [b]{\footnotesize $\rho = (n-m)/n$ }
        \item [c]{\footnotesize Maximum $k$ s.t. $\alpha_k < 1/2$ }
    \end{tablenotes}

\end{threeparttable}
\end{table*}

%%%%%%%%%%%%%%%%%%%%%%%%%%%%%%%%%%%%%%%%%%%%%%%%%
%\FloatBarrier
\begin{table*}[ht]
\centering

\begin{threeparttable}
\caption{ Exact $\alpha_k$ from the sandwiching algorithm (Fourier Matrix)}

\label{Table_Fourier_sandwiching}
\setlength{\tabcolsep}{15pt}

{\small
    \begin{tabular}{cccccccc}
        \multicolumn{8}{r} {\scriptsize (Rounded off to the nearest hundredth) }\\
        \hline \hline
        A matrix((n-m) $\times$ n)\tnote{a} &$\rho$\tnote{b} & $\alpha_1$    & $\alpha_2$\tnote{d}    & $\alpha_3$   & $\alpha_4$   & $\alpha_5$ & $k$\tnote{c} \\
        \hline
        20 $\times$ 40         &0.5        & 0.19     & 0.35     & 0.45    & 0.58         & 0.67\tnote{e} &  3 \\
        24 $\times$ 40         &0.6        & 0.18     & 0.33     & 0.47    & 0.59         & 0.70          &  3 \\
        28 $\times$ 40         &0.7        & 0.13     & 0.25     & 0.38    & 0.50         & 0.63          &  3 \\
        32 $\times$ 40         &0.8        & 0.09     & 0.17     & 0.24    & 0.31         & 0.38\tnote{e} &  5 \\
        \hline
    \end{tabular}
}
    \begin{tablenotes}
%        \item {\footnotesize Rounded off to the nearest hundredth }
        \item [a]{\footnotesize $Az = 0$ }
        \item[b] {\footnotesize $\rho = (n-m)/n$ }
        \item[c] {\footnotesize Maximum $k$ s.t. $\alpha_k < 1/2$ }
        \item[d] {\footnotesize Obtained from the sandwiching algorithm based on the pick-$1$-element algorithm }
        \item[e] {\footnotesize Obtained from the sandwiching algorithm based on the pick-$3$-element algorithm }
        \\
    \end{tablenotes}

\end{threeparttable}
\end{table*}

%%%%%%%%%%%%%%%%%%%%%%%%%%%%%%%%%%%%%%%%%%%%%%%%%
\FloatBarrier
\begin{table*}[ht]
\centering

\begin{threeparttable}
\caption{ Running steps in the sandwiching algorithm (Fourier Matrix)}

\label{Table_Fourier_sandwiching_step}
\setlength{\tabcolsep}{15pt}

{\small
    \begin{tabular}{ccccccc}
        \multicolumn{7}{r} {\scriptsize }\\
        \hline \hline
        A matrix((n-m) $\times$ n) &$\rho$\tnote{a}     & $k=1$\tnote{b}  & $k=2$\tnote{c}     & $k=3$      & $k=4$         & $k=5$     \\
        \hline
        Exhaustive Search      &-              & -          & 780     & 9880      & 91390      & 658008  \\
        20 $\times$ 40         &0.5            & -          & 780     & 720       & 3250       & 640\tnote{d}  \\
        24 $\times$ 40         &0.6            & -          & 780     & 40        & 120        & 920  \\
        28 $\times$ 40         &0.7            & -          & 175     & 120       & 270        & 280    \\
        32 $\times$ 40         &0.8            & -          & 780     & 240       & 3120       & 720\tnote{d}     \\
        \hline
    \end{tabular}
}
    \begin{tablenotes}
        \item[a]{\footnotesize $\rho = (n-m)/n$ }
        \item[b]{\footnotesize Sandwiching algorithm is not applied }
        \item[c]{\footnotesize Obtained from the sandwiching algorithm based on the pick-$1$-element algorithm }
        \item[d]{\footnotesize Obtained from the sandwiching algorithm based on the pick-$3$-element algorithm }
    \end{tablenotes}

\end{threeparttable}

\end{table*}

%\FloatBarrier
\begin{table*}[!ht]
\centering

\begin{threeparttable}
\caption{ Running time of the sandwiching algorithm (Fourier Matrix)}

\label{Table_Fourier_sandwiching_time}
\setlength{\tabcolsep}{15pt}
{\small
    \begin{tabular}{ccccccc}
        \multicolumn{7}{r} {\scriptsize (Unit: minute)}\\
        \hline \hline
        A matrix((n-m) $\times$ n) &$\rho$\tnote{a} & $k=1$     & $k=2$\tnote{b}        & $k=3$      & $k=4$         & $k=5$     \\
        \hline
        20 $\times$ 40         &0.5        & 0.10    & 8.63       & 10.95        & 13.06         & 111.97\tnote{c} \\
        24 $\times$ 40         &0.6        & 0.10    & 8.65       & 3.75         & 4.33          & 7.32 \\
        28 $\times$ 40         &0.7        & 0.18    & 2.03       & 4.64         & 4.02          & 9.17 \\
        32 $\times$ 40         &0.8        & 0.13    & 8.69       & 5.93         & 25.40         & 107.58\tnote{c}   \\
        \hline
    \end{tabular}
}
    \begin{tablenotes}
%        \item {\footnotesize Rounded off to the nearest hundredth (Unit: minute)}
        \item[a]{\footnotesize $\rho = (n-m)/n$ }
        \item[b]{\footnotesize Obtained from the sandwiching algorithm based on the pick-$1$-element algorithm }
        \item[c]{\footnotesize Obtained from the sandwiching algorithm based on the pick-$3$-element algorithm }
    \end{tablenotes}

\end{threeparttable}
\end{table*}

%%%%%%%%%%%%%%%%%%%%%%%%%%%%%%%%%%%%%%%%%%%%%%%%%
%\FloatBarrier
\begin{table*}[h]
\centering

\begin{threeparttable}
\caption{ Upper and lower bounds when $n=40$ from \cite{d2011testing} and \cite{juditsky2011verifiable} (Fourier Matrix)}

\label{Table_Fourier_cite}
\setlength{\tabcolsep}{15pt}

{\small
    \begin{tabular}{cccccccc}
        \multicolumn{7}{r} {\scriptsize }\\
        \hline \hline
        Relaxation &$\rho$     &$\alpha_1$     &$\alpha_2$    & $\alpha_3$   & $\alpha_4$   & $\alpha_5$ & $k$\tnote{c} \\
        \hline
        LP\tnote{a}   &0.5     & 0.21     & 0.38    & 0.57    & 0.82    & 0.98  & 2  \\
        SDP\tnote{b} &0.5      & 0.21     & 0.38    & 0.57    & 0.82    & 0.98  & 2  \\
        SDP low.   &0.5        & 0.05     & 0.10    & 0.16    & 0.24    & 0.32  & 2  \\
        \hline
        LP         &0.6        & 0.16     & 0.31    & 0.46    & 0.61    & 0.82  & 3  \\
        SDP        &0.6        & 0.16     & 0.31    & 0.46    & 0.61    & 0.82  & 3  \\
        SDP low.   &0.6        & 0.04     & 0.09    & 0.15    & 0.20    & 0.31  & 3  \\
        \hline
        LP         &0.7        & 0.12     & 0.25    & 0.39    & 0.50    & 0.62  & 3  \\
        SDP        &0.7        & 0.12     & 0.25    & 0.39    & 0.50    & 0.62  & 3  \\
        SDP low.   &0.7        & 0.04     & 0.09    & 0.14    & 0.18    & 0.22  & 3  \\
        \hline
        LP         &0.8        & 0.10     & 0.20    & 0.30    & 0.38    & 0.48  &  5 \\
        SDP        &0.8        & 0.10     & 0.20    & 0.30    & 0.38    & 0.48  &  5 \\
        SDP low.   &0.8        & 0.04     & 0.07    & 0.13    & 0.17    & 0.23  &  5 \\
        \hline
    \end{tabular}
}
    \begin{tablenotes}
        \item[a] {\footnotesize Linear Programming }
        \item[b] {\footnotesize Semidefinite Programming }
        \item[c] {\footnotesize Maximum $k$ s.t. $\alpha_k < 1/2$ }
    \end{tablenotes}

\end{threeparttable}
\end{table*}

\FloatBarrier

%%%%%%%%%%%%%%%%%%%%%%%%%%%%%%%%%%%%%%%%%%%%%%%%%%%%%%%%%%%%%%%%% Fourier Matrix
\newpage
%%%%%%%%%%%%%%%%%%%%%%%%%%%%%%%%%%%%%%%%%%%%%%%%%
%\FloatBarrier
\begin{table*}[!ht]
\centering

\begin{threeparttable}[t]
\caption{ Upper bounds from the pick-$1$-element algorithm (Bernoulli Matrix)}

\label{Table_Bernoulli_pick1}
\setlength{\tabcolsep}{15pt}

{\small
    \begin{tabular}{cccccccc}
        \multicolumn{8}{r} {\scriptsize (Rounded off to the nearest hundredth) }\\
        \hline \hline
        A matrix((n-m) $\times$ n)\tnote{a} &{$\rho$}\tnote{b}  & $\alpha_1$    & $\alpha_2$    & $\alpha_3$   & $\alpha_4$   & $\alpha_5$ & $k$\tnote{c}  \\
        \hline
        20 $\times$ 40         &0.5        &0.25      &0.49      &0.72     &0.95     &1.17   &2  \\
        24 $\times$ 40         &0.6        &0.22      &0.41      &0.60     &0.79     &0.97   &2  \\
        28 $\times$ 40         &0.7        &0.19      &0.36      &0.53     &0.68     &0.83   &2  \\
        32 $\times$ 40         &0.8        &0.14      &0.28      &0.41     &0.54     &0.66   &3  \\
        \hline
    \end{tabular}
}
\begin{tablenotes}
%    \item {\footnotesize Rounded off to the nearest hundredth }
    \item [a]{\footnotesize $Az = 0$ }
    \item [b]{\footnotesize $\rho = (n-m)/n$ }
    \item [c]{\footnotesize Maximum $k$ s.t. $\alpha_k < 1/2$ }
\end{tablenotes}

\end{threeparttable}
\end{table*}

%%%%%%%%%%%%%%%%%%%%%%%%%%%%%%%%%%%%%%%%%%%%%%%%%
%\FloatBarrier
\begin{table*}[ht]
\centering

\begin{threeparttable}
\caption{ Upper bounds from the pick-$2$-element algorithm (Bernoulli Matrix)}

\label{Table_Bernoulli_pick2}
\setlength{\tabcolsep}{15pt}

{\small
    \begin{tabular}{cccccccc}
        \multicolumn{8}{r} {\scriptsize (Rounded off to the nearest hundredth) }\\
        \hline \hline
        A matrix((n-m) $\times$ n)\tnote{a} &$\rho$\tnote{b} & $\alpha_1$    & $\alpha_2$    & $\alpha_3$   & $\alpha_4$   & $\alpha_5$ & $k$\tnote{c} \\
        \hline
        20 $\times$ 40          &0.5        &0.25      &0.42      &0.60     &0.78     &0.96   & 2  \\
        24 $\times$ 40          &0.6        &0.22      &0.36      &0.51     &0.67     &0.82   & 2  \\
        28 $\times$ 40          &0.7        &0.19      &0.29      &0.43     &0.55     &0.67   & 3  \\
        32 $\times$ 40          &0.8        &0.14      &0.24      &0.35     &0.45     &0.55   & 4  \\
        \hline
    \end{tabular}
}
\begin{tablenotes}
%    \item {\footnotesize Rounded off to the nearest hundredth }
    \item [a]{\footnotesize $Az = 0$ }
    \item [b]{\footnotesize $\rho = (n-m)/n$ }
    \item [c]{\footnotesize Maximum $k$ s.t. $\alpha_k < 1/2$ }
\end{tablenotes}
\end{threeparttable}
\end{table*}

%%%%%%%%%%%%%%%%%%%%%%%%%%%%%%%%%%%%%%%%%%%%%%%%%
%\FloatBarrier
\begin{table*}[ht]
\centering

\begin{threeparttable}
\caption{ Upper bounds from the pick-$3$-element algorithm (Bernoulli Matrix)}

\label{Table_Bernoulli_pick3}
\setlength{\tabcolsep}{15pt}

{\small
    \begin{tabular}{cccccccc}
        \multicolumn{8}{r} {\scriptsize (Rounded off to the nearest hundredth) }\\
        \hline \hline
        A matrix((n-m) $\times$ n)\tnote{a} &$\rho$\tnote{b} & $\alpha_1$    & $\alpha_2$    & $\alpha_3$   & $\alpha_4$   & $\alpha_5$ & $k$\tnote{c} \\
        \hline
        20$\times$ 40          &0.5        &0.25      &0.42      &0.53     &0.69     &0.85   &2   \\
        24 $\times$ 40         &0.6        &0.22      &0.36      &0.46     &0.60     &0.73   &3   \\
        28 $\times$ 40         &0.7        &0.19      &0.29      &0.39     &0.51     &0.62   &3   \\
        32 $\times$ 40         &0.8        &0.14      &0.24      &0.31     &0.41     &0.50   &5   \\
        \hline
    \end{tabular}
}
    \begin{tablenotes}
%        \item {\footnotesize Rounded off to the nearest hundredth }
        \item [a]{\footnotesize $Az = 0$ }
        \item [b]{\footnotesize $\rho = (n-m)/n$ }
        \item [c]{\footnotesize Maximum $k$ s.t. $\alpha_k < 1/2$ }
    \end{tablenotes}

\end{threeparttable}
\end{table*}

%%%%%%%%%%%%%%%%%%%%%%%%%%%%%%%%%%%%%%%%%%%%%%%%%
%\FloatBarrier
\begin{table*}[ht]
\centering

\begin{threeparttable}
\caption{ Exact $\alpha_k$ from the sandwiching algorithm (Bernoulli Matrix)}

\label{Table_Bernoulli_sandwiching}
\setlength{\tabcolsep}{15pt}

{\small
    \begin{tabular}{cccccccc}
        \multicolumn{8}{r} {\scriptsize (Rounded off to the nearest hundredth) }\\
        \hline \hline
        A matrix((n-m) $\times$ n)\tnote{a} &$\rho$\tnote{b} & $\alpha_1$    & $\alpha_2$\tnote{d}    & $\alpha_3$   & $\alpha_4$   & $\alpha_5$ & $k$\tnote{c} \\
        \hline
        20 $\times$ 40         &0.5        & 0.25     & 0.41     & 0.52    & 0.62         & 0.70\tnote{e} &  2  \\
        24 $\times$ 40         &0.6        & 0.23     & 0.35     & 0.45    & 0.56         & 0.65\tnote{e} &   3  \\
        28 $\times$ 40         &0.7        & 0.17     & 0.30     & 0.39    & 0.47         & 0.54\tnote{e} &   4  \\
        32 $\times$ 40         &0.8        & 0.14     & 0.24     & 0.32    & 0.40         & 0.46       &   5  \\
        \hline
    \end{tabular}
}
    \begin{tablenotes}
%        \item {\footnotesize Rounded off to the nearest hundredth }
        \item[a] {\footnotesize $Az = 0$ }
        \item[b] {\footnotesize $\rho = (n-m)/n$ }
        \item[c] {\footnotesize Maximum $k$ s.t. $\alpha_k < 1/2$ }
        \item[d] {\footnotesize Obtained from the sandwiching algorithm based on the pick-$1$-element algorithm }
        \item[e] {\footnotesize Obtained from the sandwiching algorithm based on the pick-$3$-element algorithm }
        \\
    \end{tablenotes}

\end{threeparttable}
\end{table*}

%%%%%%%%%%%%%%%%%%%%%%%%%%%%%%%%%%%%%%%%%%%%%%%%%
\FloatBarrier
\begin{table*}[ht]
\centering

\begin{threeparttable}
\caption{ Running steps in the sandwiching algorithm (Bernoulli Matrix)}

\label{Table_Bernoulli_sandwiching_step}
\setlength{\tabcolsep}{15pt}

{\small
    \begin{tabular}{ccccccc}
        \multicolumn{7}{r} {\scriptsize }\\
        \hline \hline
        A matrix((n-m) $\times$ n) &$\rho$\tnote{a}     & $k=1$\tnote{b}  & $k=2$\tnote{c}     & $k=3$      & $k=4$         & $k=5$     \\
        \hline
        Exhaustive Search      &-              & -          & 780     & 9880    & 91390       & 658008        \\
        20 $\times$ 40         &0.5            & -          & 218     & 63      & 8789        & 1004\tnote{d} \\
        24 $\times$ 40         &0.6            & -          & 124     & 36      & 809         & 27\tnote{d}   \\
        28 $\times$ 40         &0.7            & -          & 59      & 7       & 231         & 36\tnote{d}   \\
        32 $\times$ 40         &0.8            & -          & 33      & 5       & 66          & 2303          \\
        \hline
    \end{tabular}
}
    \begin{tablenotes}
        \item[a]{\footnotesize $\rho = (n-m)/n$ }
        \item[b]{\footnotesize Sandwiching algorithm is not applied }
        \item[c]{\footnotesize Obtained from the sandwiching algorithm based on the pick-$1$-element algorithm }
        \item[d]{\footnotesize Obtained from the sandwiching algorithm based on the pick-$3$-element algorithm }
    \end{tablenotes}

\end{threeparttable}

\end{table*}

%\FloatBarrier
\begin{table*}[!ht]
\centering

\begin{threeparttable}
\caption{ Running time of the sandwiching algorithm (Bernoulli Matrix)}

\label{Table_Bernoulli_sandwiching_time}
\setlength{\tabcolsep}{15pt}
{\small
    \begin{tabular}{ccccccc}
        \multicolumn{7}{r} {\scriptsize (Unit: minute)}\\
        \hline \hline
        A matrix((n-m) $\times$ n) &$\rho$\tnote{a} & $k=1$     & $k=2$\tnote{b}        & $k=3$      & $k=4$         & $k=5$     \\
        \hline
        20 $\times$ 40         &0.5        & 0.10    & 2.55       & 4.15        & 57.93        & 98.87\tnote{c}  \\
        24 $\times$ 40         &0.6        & 0.10    & 1.51       & 3.88        & 7.60         & 93.09\tnote{c}  \\
        28 $\times$ 40         &0.7        & 0.11    & 0.79       & 3.65        & 4.59         & 92.07\tnote{c}  \\
        32 $\times$ 40         &0.8        & 0.11    & 0.50       & 3.55        & 3.99         & 16.89           \\
        \hline
    \end{tabular}
}
    \begin{tablenotes}
%        \item {\footnotesize Rounded off to the nearest hundredth (Unit: minute)}
        \item[a]{\footnotesize $\rho = (n-m)/n$ }
        \item[b]{\footnotesize Obtained from the sandwiching algorithm based on the pick-$1$-element algorithm }
        \item[c]{\footnotesize Obtained from the sandwiching algorithm based on the pick-$3$-element algorithm }
    \end{tablenotes}

\end{threeparttable}
\end{table*}

%%%%%%%%%%%%%%%%%%%%%%%%%%%%%%%%%%%%%%%%%%%%%%%%%
%\FloatBarrier
\begin{table*}[h]
\centering

\begin{threeparttable}
\caption{ Upper and lower bounds when $n=40$ from \cite{d2011testing} and \cite{juditsky2011verifiable} (Bernoulli Matrix)}

\label{Table_Bernoulli_cite}
\setlength{\tabcolsep}{15pt}

{\small
    \begin{tabular}{cccccccc}
        \multicolumn{7}{r} {\scriptsize }\\
        \hline \hline
        Relaxation &$\rho$     &$\alpha_1$     &$\alpha_2$    & $\alpha_3$   & $\alpha_4$   & $\alpha_5$ & $k$\tnote{c} \\
        \hline
        LP\tnote{a}   &0.5     & 0.25     & 0.45    & 0.64    & 0.82    & 0.97  & 2  \\
        SDP\tnote{b} &0.5      & 0.25     & 0.45    & 0.63    & 0.80    & 0.94  & 2  \\
        SDP low.   &0.5        & 0.25     & 0.28    & 0.29    & 0.29    & 0.34  & 2  \\
        \hline
        LP         &0.6        & 0.21     & 0.38    & 0.55    & 0.69    & 0.83  & 2  \\
        SDP        &0.6        & 0.21     & 0.38    & 0.54    & 0.68    & 0.81  & 2  \\
        SDP low.   &0.6        & 0.21     & 0.26    & 0.29    & 0.33    & 0.34  & 2  \\
        \hline
        LP         &0.7        & 0.17     & 0.32    & 0.46    & 0.58    & 0.70  & 3  \\
        SDP        &0.7        & 0.17     & 0.32    & 0.46    & 0.58    & 0.69  & 3  \\
        SDP low.   &0.7        & 0.17     & 0.24    & 0.29    & 0.33    & 0.37  & 3  \\
        \hline
        LP         &0.8        & 0.14     & 0.26    & 0.38    & 0.47    & 0.57  & 4  \\
        SDP        &0.8        & 0.14     & 0.26    & 0.37    & 0.47    & 0.57  & 4  \\
        SDP low.   &0.8        & 0.14     & 0.21    & 0.27    & 0.33    & 0.38  & 4  \\
        \hline
    \end{tabular}
}
    \begin{tablenotes}
        \item[a] {\footnotesize Linear Programming }
        \item[b] {\footnotesize Semidefinite Programming }
        \item[c] {\footnotesize Maximum $k$ s.t. $\alpha_k < 1/2$ }
    \end{tablenotes}

\end{threeparttable}
\end{table*}

%%%%%%%%%%%%%%%%%%%%%%%%%%%%%%%%%%%%%%%%%%%%%%%%%%%%%%%%%%%%%%%%%

%\begin{figure}[!h]
%    \caption{Global upper Bound (GUB) and lower bound (GLB) in the sandwiching algorithm based on the pick-$2$-element algorithm ($\alpha_5$ on $40 \times 20$ $H$ matrix case) }
%    \label{upper_lower_graph}
%    \centering
%    \includegraphics[scale=0.5]{img/upper_lower_graph.eps}
%\end{figure}

%\section{Conclusion}
%\label{Sec6}
%
%In this paper, we proposed new algorithms to verify the null space conditions.
%We first proposed a series of new polynomial-time algorithms to compute upper bounds on $\alpha_k$.
%Based on these new polynomial-time algorithms, we further designed a new sandwiching algorithm, to compute the \emph{exact} $\alpha_k$ with greatly reduced complexity.
%
%The future work for verifying the null space conditions includes designing efficient algorithms to reduce the operation time even more.

%\bibliographystyle{IEEEbib}
%\bibliography{refs}

\end{document}